\documentclass[reqno,11pt]{amsart}
\usepackage{amsmath, latexsym, amsfonts, amssymb, amsthm, amscd, mathabx}
\usepackage{graphics,epsf,psfrag}
\numberwithin{equation}{section}

\newtheorem{theorem}{Theorem}[section]
\newtheorem{lemma}[theorem]{Lemma}
\newtheorem{proposition}[theorem]{Proposition}

\newtheorem{rem}[theorem]{Remark}

\newcommand{\ind}{\mathbf{1}}

\newcommand{\R}{\mathbb{R}}
\newcommand{\Z}{\mathbb{Z}}
\newcommand{\N}{\mathbb{N}}
\renewcommand{\tilde}{\widetilde}
\renewcommand{\hat}{\widehat}
\renewcommand{\bar}{\widebar}
\renewcommand{\check}{\widecheck}

\DeclareMathSymbol{\leqslant}{\mathalpha}{AMSa}{"36} 
\DeclareMathSymbol{\geqslant}{\mathalpha}{AMSa}{"3E} 
\DeclareMathSymbol{\eset}{\mathalpha}{AMSb}{"3F}     
\newcommand{\dd}{\,\text{\rm d}}             

\newcommand{\sumtwo}[2]{\sum_{\substack{#1 \\ #2}}} 
\DeclareMathOperator{\var}{\mathrm{Var}}

\newcommand{\bbP}{{\ensuremath{\mathbb P}} }

\newcommand{\gb}{\beta}
\newcommand{\gga}{\gamma}            
\newcommand{\gd}{\delta}
\newcommand{\gep}{\varepsilon}       

\newcommand{\go}{\omega}

\newcommand{\gl}{\lambda}

\makeatletter
\def\captionfont@{\footnotesize}
\def\captionheadfont@{\scshape}

\long\def\@makecaption#1#2{%
  \vspace{2mm}
  \setbox\@tempboxa\vbox{\color@setgroup
    \advance\hsize-6pc\noindent
    \captionfont@\captionheadfont@#1\@xp\@ifnotempty\@xp
        {\@cdr#2\@nil}{.\captionfont@\upshape\enspace#2}%
    \unskip\kern-6pc\par
    \global\setbox\@ne\lastbox\color@endgroup}%
  \ifhbox\@ne 
    \setbox\@ne\hbox{\unhbox\@ne\unskip\unskip\unpenalty\unkern}%
  \fi
  \ifdim\wd\@tempboxa=\z@ 
    \setbox\@ne\hbox to\columnwidth{\hss\kern-6pc\box\@ne\hss}%
  \else 
    \setbox\@ne\vbox{\unvbox\@tempboxa\parskip\z@skip
        \noindent\unhbox\@ne\advance\hsize-6pc\par}%
\fi
  \ifnum\@tempcnta<64 
    \addvspace\abovecaptionskip
    \moveright 3pc\box\@ne
  \else 
    \moveright 3pc\box\@ne
    \nobreak
    \vskip\belowcaptionskip
  \fi
\relax
}
\makeatother
\def\writefig#1 #2 #3 {\rlap{\kern #1 truecm
\raise #2 truecm \hbox{#3}}}

\newcommand{\tf}{\textsc{f}}

\newcommand{\tr}{\text{Trace}}
\title[Free energy of directed polymers in dimension $1+1$ and $1+2$]{New bounds for the free energy of directed polymers in dimension $1+1$ and $1+2$}

\author{Hubert Lacoin}
\address{
  Universit{\'e} Paris Diderot and Laboratoire de Probabilit{\'e}s et Mod\`eles Al\'eatoires (CNRS U.M.R. 7599),
U.F.R.                Math\'ematiques, Case 7012 (Site Chevaleret),
                75205 Paris cedex 13, France
}
\email{lacoin\@@math.jussieu.fr}

\begin{document}
\maketitle
\begin{abstract}
We study the free energy of the directed polymer in random environment model in dimension $1+1$ and $1+2$.
For dimension one, we improve the statement of Comets and Vargas in \cite{CV} concerning very strong disorder by giving sharp estimates on the free energy at high temperature. In dimension two, we prove that very strong disorder holds at all temperatures, thus solving a long standing conjecture in the field.
\end{abstract}

\section{Introduction}
\subsection{The model}
We study a directed polymer model introduced by Huse and Henley (in dimension $1+1$) \cite{HH} with the purpose of investigating  impurity-induced domain-wall roughening in the 2D-Ising model. The first mathematical study of directed polymers in random environment was made by Imbrie and Spencer \cite{IS}, and was followed by numerous authors \cite{IS, B, AZ, SZ, CH_ptrf, CSY, CY, CH_al, CV, V} (for a review on the subject see \cite{CSY_rev}). 
Directed polymers in random environment model, in particular, polymer chains in a solution with impurities.

In our set--up the polymer chain is the graph $\{(i,S_i)\}_{1\le i\le N}$ of a nearest--neighbor path in $\Z^d$, $S$ starting from zero.
 The equilibrium behavior of this chain is described by a measure on the set of paths:
the impurities enter the definition of the measure as {\sl disordered potentials}, given by a typical realization of a field of i.i.d.\ random variables $\go=\{\go_{(i,z)}\ ;\ i\in \N, z\in \Z^d\}$ (with associated law $Q$). The polymer chain will tend to be attracted by larger values of the environment and repelled by smaller ones.
More precisely, we define the Hamiltonian 
\begin{equation}
 H_N(S):=\sum_{i=1}^N \go_{i,S_i}.
\end{equation}
We denote by $P$ the law of the simple symmetric random walk on $\Z^d$ starting at $0$ (in the sequel $P f(S)$, respectively $Q g(\go)$, will denote the expectation with respect to $P$, respectively Q).
One defines the polymer measure of order $N$ at inverse temperature $\gb$ as
\begin{equation}
 \mu^{(\gb)}_N(S)=\mu_N(S):=\frac{1}{Z_N}\exp\left(\gb H_N(S)\right)P(S),
\end{equation}
where $Z_N$ is the normalization factor which makes $\mu_N$ a probability measure
\begin{equation}
 Z_N:=P \exp\left(\gb H_N(S)\right).
\end{equation}
We call $Z_N$ the {\sl partition function} of the system. In the sequel, we will consider the case of $\go_{(i,z)}$ with zero mean and unit variance and such that there exists $B \in(0,\infty]$ such that
\begin{equation}\label{expm}
 \gl(\gb)=\log Q \exp(\gb\go_{(1,0)})<\infty, \quad \text{ for } 0\le \gb \le B. 
\end{equation}
Finite exponential moments are required to guarantee that $Q Z_N<\infty$.
The model can be defined and it is of interest also with environments with heavier tails (see e.g.\ \cite{V}) but we will not consider these cases here.

\subsection{Weak, strong and very strong disorder}

In order to understand the role of disorder in the behavior of $\mu_N$, as $N$ becomes large, let us observe that, when $\gb=0$, $\mu_N$ is the law of the simple random walk, so that we know that, properly rescaled, the polymer chain will look like the graph of a $d$-dimensional Brownian motion. The main questions that arise for our model for $\gb>0$ are whether or not the presence of disorder breaks the diffusive behavior of the chain for large $N$, and how the polymer measure looks like when diffusivity does not hold.
\medskip

Many authors have studied diffusivity in polymer models: in \cite{B}, Bolthausen remarked that the renormalized partition function $W_N:=Z_N/(Q Z_N)$ has a martingale property and proved the following zero-one law:
\begin{equation}
 Q\left\{\lim_{N\to\infty} W_N=0\right\}\in\{0,1\}.
\end{equation}
A series of paper  \cite{IS,B,AZ,SZ,CY} lead to
\begin{equation}
  Q\left\{\lim_{N\to\infty} W_N=0\right\}=0 \Rightarrow  \text{ diffusivity },
\end{equation}
and a consensus in saying that this implication is an equivalence.
For this reason, it is natural and it has become customary to say that {\sl weak disorder} holds when $W_N$ converges to some non-degenerate limit and that {\sl strong disorder} holds when $W_N$ tends to zero.
\medskip

Carmona and Hu \cite{CH_ptrf} and Comets, Shiga and Yoshida \cite{CSY} \ proved that strong disorder holds for all $\gb$ in dimension $1$ and $2$. The result was completed by Comets and Yoshida \cite{CY}: we summarize it here

\medskip
\begin{theorem}\label{strdis}
 There exists a critical value $\gb_c=\gb_c(d)\in[0,\infty]$ (depending of the law of the environment) such that
\begin{itemize}
 \item Weak disorder holds when $\gb<\gb_c$.
 \item Strong disorder holds when $\gb>\gb_c$.
\end{itemize}
Moreover:
\begin{equation}\begin{split}
 \gb_c(d)&=0 \text{ for } d=1,2\\
 \gb_c(d)&\in (0,\infty] \text{ for } d\ge 3. 
\end{split}
\end{equation}
\end{theorem}
We mention also that the case $\gb_c(d)=\infty$ can only occur when the random variable $\go_{(0,1)}$ is bounded.
\medskip

In \cite{CH_ptrf} and \cite{CSY} a characterization of strong disorder has been obtained in term of localization of the polymer chain: we cite the following result \cite[Theorem 2.1]{CSY}
\medskip

\begin{theorem}
If $S^{(1)}$ and $S^{(2)}$ are two i.i.d.\ polymer chains, we have
\begin{equation}
 Q\left\{ \lim_{N\rightarrow\infty} W_N=0 \right\}=Q\left\{\sum_{N\ge 1} \mu_{N-1}^{\otimes 2}(S_N^{(1)}=S_N^{(2)})=\infty \right\}
\end{equation}
Moreover if   $Q\{ \lim_{N\rightarrow\infty} W_N=0 \}=1$ there exists a constant $c$ (depending on $\gb$ and the law of the environment) such that for 
\begin{equation}\label{localization}
-c\log W_N \le \sum_{n= 1}^N \mu_{n-1}^{\otimes 2}(S_n^{(1)}=S_n^{(2)})\le -\frac{1}{c} \log W_N.
\end{equation}
\end{theorem}
\medskip
One can notice that \eqref{localization} has a very strong meaning in term of trajectory localization when $W_N$ decays exponentially:
it implies that two independent polymer chains tend to share the same endpoint with positive probability.
For this reason we introduce now the notion of free energy, we refer to \cite[Proposition 2.5]{CSY} and \cite[Theorem 
3.2]{CY} for the following result:
\medskip
\begin{proposition}
 The quantity
\begin{equation}
 p(\gb):=\lim_{N\to\infty}\frac{1}{N}\log W_N \ ,
\end{equation}
	exists $Q$-a.s., it is non-positive and non-random. We call it the {\sl free energy} of the model,
 and we have
\begin{equation}\label{Pn}
 p(\gb)=\lim_{N\to\infty}\frac{1}{N}Q\log W_N=:\lim_{N\to\infty} p_N(\gb).
\end{equation}
Moreover $p(\gb)$ is non-increasing in $\gb$.
\end{proposition}
\medskip

We stress that the inequality $p(\gb)\le 0$ is the standard {\sl annealing}  bound. In view on \eqref{localization}, it is natural to say that {\sl very strong disorder} holds whenever $p(\gb)<0$.
One can moreover define $\bar \gb_c(d)$ the critical value of $\gb$ for the free energy i.e.\ :
\begin{equation}
 p(\gb)<0 \Leftrightarrow \gb>\bar\gb_c(d).
\end{equation}
Let us stress that, from the physicists' viewpoint, $\bar\gb_c(d)$ is the natural critical point because it is a point of non-analyticity of the free energy (at least if $\bar \gb_c(d)>0$). In view of this definition, we obviously have $\bar\gb_c(d)\ge \gb_c(d)$. 
It is widely believed that $\bar\gb_c(d)=\gb_c(d)$, i.e.\ that there exists no intermediate phase where we have {\sl strong disorder} but not {\sl very strong disorder}. However, this is a challenging question:
Comets and Vargas \cite{CV} answered it in dimension $1+1$ by proving that $\bar\gb_c(1)=0$. In this paper, we make their result more precise. Moreover we prove that $\bar \gb_c(2)=0$.

\subsection{Presentation of the results}

The first aim of this paper is to sharpen the result of Comets and Vargas on the $1+1$-dimensional case. In fact, we are going
to give a precise statement on the behavior of $p(\gb)$ for small $\gb$. Our result is the following
\medskip
\begin{theorem}\label{pasgaussi}
When $d=1$ and the environment satisfies \eqref{expm}, there exist constants $c$ and $\gb_0< B$ (depending on the distribution of the environment) such that for all $0\le \gb\le \gb_0$ we have
\begin{equation}
-\frac{1}{c} \gb^4[1+(\log \gb)^2] \le p(\gb)\le -c \gb^4.
\end{equation}
\end{theorem}
\medskip
We believe that the logarithmic factor in the lower bound is an artifact of the method. In fact, by using replica-coupling, we have been able to get rid of it in the Gaussian case.
\medskip
\begin{theorem}\label{gaussi}
When $d=1$ and the environment is Gaussian, there exists a constant $c$ such that for all $\gb \le 1$.
\begin{equation}
-\frac{1}{c}\gb^4 \le p(\gb)\le  -c \gb^4.
\end{equation}
\end{theorem}
\medskip
These estimates concerning the free energy give us some idea of the behavior of $\mu_N$ for small $\gb$. Indeed, Carmona and Hu in \cite[Section 7]{CH_ptrf} proved a relation between $p(\gb)$ and the overlap (although their notation differs from ours).
This relation together with our estimates for $p(\gb)$ suggests that, for low $\gb$, the asymptotic contact fraction between independent polymers
\begin{equation}
 \lim_{N\to\infty} \frac{1}{N}\mu_N^{\otimes 2} \sum_{n=1}^N \ind_{\{S_n^{(1)}=S_n^{(2)}\}},
\end{equation}
behaves like $\gb^2$.
\medskip

The second result we present is that $\bar \gb_c(2)=0$. As for the $1+1$-dimensional case, our approach yields an explicit bound on $p(\gb)$ for $\gb$ close to zero.
\medskip

\begin{theorem}\label{1+2uplb}
 When $d=2$, there exist constants $c$ and $\gb_0$ such that for all $\gb\le \gb_0$,
\begin{equation}
-\exp\left(-\frac{1}{c\gb^2}\right) \le p(\gb)\le -\exp\left(-\frac{c}{\gb^4}\right),
\end{equation}
so that
\begin{equation}
\bar \gb_c(2)=0,
\end{equation}
and $0$ is a point of non-analyticity for $p(\gb)$.
\end{theorem}
\medskip

\begin{rem}\rm After the appearance of this paper as a preprint, the proof of the above result has been adapted by Bertin \cite{Bertin} to prove the exponential decay of the partition function for {\sl Linear Stochastic Evolution} in dimension $2$, a model that is a slight generalisation of directed polymer in random environment.
\end{rem}
\medskip

\begin{rem}\rm
 Unlike in the one dimensional case, the two bounds on the free energy provided by our methods do not match. We believe that the second moment method, that gives the lower bound is quite sharp and gives the right order of magnitude for $\log p(\gb)$.
The method developped in \cite{cf:kbodies} to sharpen the estimate on the critical point shift for pinning models at marginality adapted to the context of directed polymer should be able to improve the result, getting $p(\gb)\le -\exp(-c_{\gep} \gb^{-(2+\gep)})$ for all $\gb\le 1$ for any $\gep$. 
\end{rem}

\subsection{Organization of the paper}
The various techniques we use have been inspired by ideas used successfully for another polymer model, namely the polymer pinning on a defect line (see \cite{F_rc,cf:GLT,cf:DGLT,F,GLTm}). 

However the ideas we use to establish lower bounds differ sensibly from the ones leading to the upper bounds. For this reason, we present first the proofs of the upper bound results in Section \ref{rough}, \ref{onedim} and \ref{twodim}. The lower bound results are proven in Section \ref{lb11}, \ref{dim11} and \ref{dim12}.
\medskip

To prove the lower bound results, we use a technique that combines the so-called {\sl fractional moment method} and change of measure. This approach has been first used for pinning model in \cite{cf:DGLT} and it has been
refined since in \cite{F,GLTm}. In Section \ref{rough}, we prove a non-optimal upper bound for the free energy in the case of Gaussian environment in dimension $1+1$ to introduce the reader to this method. In Section \ref{onedim} we prove the optimal upper bound for arbitrary environment in dimension $1+1$, and in Section \ref{twodim} we prove our upper bound for the free energy in dimension $1+2$ which implies that {\sl very strong disorder} holds for all $\gb$. These sections are placed in increasing order of technical complexity, and therefore, should be read in that order.
\medskip

Concerning the lower--bounds proofs: Section \ref{lb11} presents a proof of the lower bound of Theorem \ref{pasgaussi}. The proof combines the second moment method and a directed percolation argument. In Section \ref{dim11} the optimal bound is proven for Gaussian environment, with a specific Gaussian approach similar to what is done in \cite{F_rc}. In Section \ref{dim12} we prove the lower bound for arbitrary environment in dimension $1+2$.
These three parts are completely independent of each other.

\section{Some warm up computations}\label{rough}

\subsection{Fractional moment}

Before going into the core of the proof, we want to present here the starting step that will be used repeatedly thourough Sections \ref{rough}, \ref{onedim} and \ref{twodim}.
We want to find an upper--bound for the quantity
\begin{equation}
 p(\gb)=\lim_{N\to \infty}\frac{1}{N} Q \log W_N. 
\end{equation}
However, it is not easy to handle the expectation of a $\log$, for this reason we will use the following trick . Let $\theta\in(0,1)$, we have (by Jensen inequality)

\begin{equation}
Q \log W_N=\frac{1}{\theta}Q\log W_N^{\theta}\le \frac{1}{\theta}\log Q W_N^{\theta}.
\end{equation}
Hence

\begin{equation}\label{momentfrac}
 p(\gb)\le \liminf_{N\to\infty} \frac{1}{\theta N}\log Q W_N^{\theta}.
\end{equation}
We are left with showing that the fractional moment $Q W_N^{\theta}$ decays exponentially which is a problem that is easier to handle.

\subsection{A non optimal upper--bound in dimension $1+1$}

To introduce the reader to the general method used in this paper, combining fractional moment and change of measure, we start by proving a non--optimal result for the free--energy, using a finite volume criterion. As a more complete result is to be proved in the next section, we restrict to the Gaussian case here. The method used here is based on the one  of \cite{CV}, marorizing the free energy of the directed polymer by the one of multiplicative cascades. Let us mention that is has bee shown recently by Liu and Watbled \cite{LV} that this majoration is in a sense optimal, they obtained this result by improving the concentration inequality for the free energy.
\medskip

The idea of combining fractional moment with change of measure and finite volume criterion has been used with success for the pinning model in \cite{cf:DGLT}.

 \begin{proposition}\label{lb}
 There exists a constant $c$ such that for all $\gb\le 1$ 

\begin{equation}
 p(\gb)\le -\frac{ c\gb^4}{(|\log \gb|+1)^2}
\end{equation}
 \end{proposition}

\begin{proof}[Proof of Proposition \ref{lb} in the case of Gaussian environment]
For $\beta$ sufficiently small, we choose $n$ to be equal to $\left\lceil \frac{C_1 |\log \gb|^{2}}{\gb^4}\right\rceil$ for a fixed constant $C_1$ (here and thourough the paper for $x\in\R$, $\lceil x\rceil$, respectively $\lfloor x\rfloor$ will denote the upper, respectively the lower integer part of $x$) and define $\theta:=1-(\log n)^{-1}$.
For $x\in \Z$ we define
\begin{equation}
 W_n(x):= P  \exp\left(\sum_{i=1}^n [\gb\go_{(i,S_i)}-\gb^2/2]\right)\ind_{\{S_n=x\}}.
\end{equation}
Note that $\sum_{x\in \Z} W_n(x)=W_n$.
We use a statement which can be found in the proof of Theorem 3.3. in \cite{CV}:

\begin{equation}
\log Q W_{nm}^{\theta}\le m \log Q \sum_{x\in \Z} [W_n(x)]^{\theta} \quad \forall m\in \N.
\end{equation}
This combined with \eqref{momentfrac} implies that

\begin{equation}
p(\gb)\le \frac{1}{\theta n}\log Q \sum_{x\in \Z} [W_n(x)]^{\theta}.
\end{equation}

Hence, to prove the result, it is sufficient to show that

\begin{equation}
 Q \sum_{x\in \Z} [W_n(x)]^{\theta}\le e^{-1}, \label{eq:aprouver}
\end{equation}
for our choice of $\theta$ and $n$.

In order to estimate $Q [W_n(x)]^{\theta}$ we use an auxiliary measure $\tilde Q$.
The region where the walk $(S_i)_{0\le i\le n}$ is likely to go is
$J_n=\left([1,n]\times[-C_2\sqrt{n},C_2\sqrt{n}]\right)\cap \N\times \Z$ where $C_2$ is a big constant.

We define $\tilde Q$ as the measure under which the $\go_{i,x}$ are still independent Gaussian variables with variance $1$, but
such that $\tilde Q \go_{i,x}=-\delta_n\ind_{(i,x)\in J_n}$ where $\gd_n=1/(n^{3/4}\sqrt{2C_2\log n})$. This measure is absolutely continuous with respect to $Q$ and

\begin{equation} \label{eq:shift}
 \frac{\dd \tilde Q}{\dd Q}=\exp\left(-\sum_{(i,x)\in J_n}\left[\delta_n\go_{i,x}+\frac{\gd_n^2}{2}\right]\right).
\end{equation}
Then we have for any $x\in \Z$, using the H\"older inequality we obtain,

\begin{equation}\label{boun}
 Q \left[W_n(x)^{\theta}\right]=\tilde Q \left[\frac{\dd Q}{\dd \tilde Q}\left(W_n(x)\right)^\theta\right]\le \left(\tilde Q \left[\left(\frac{\dd Q}{\dd \tilde Q}\right)^{\frac{1}{1-\theta}}\right]\right)^{1-\theta}\left( \tilde Q W_n(x)\right)^{\theta}.
\end{equation}
The first term on the right-hand side can be computed explicitly and is equal to

\begin{equation}\label{bbboun}
\left(Q \left(\frac{\dd Q}{\dd \tilde Q}\right)^{\frac{\theta}{1-\theta}}\right)^{1-\theta}= \exp\left(\frac{\theta\gd_n^2}{2(1-\theta)}\# J_n \right)\le e,
\end{equation}
where the last inequality is obtained by replacing $\gd_n$ and $\theta$ by their values (recall $\theta=1-(\log n)^{-1}$). Therefore combining \eqref{boun} and \eqref{bbboun} we get that

\begin{equation}
  Q \sum_{x\in \Z} \left(W_n(x)\right)^{\theta}\le e \sum_{|x|\le n} \left(\tilde Q W_n(x)\right)^{\theta}.
\end{equation}

To bound the right--hand side, we first get rid of the exponent $\theta$ in the following way:

\begin{multline}
 \sum_{|x|\le n} n^{-3 \theta}\left(\tilde Q W_n(x)\right)^{\theta}\le n^{-3\theta}\#\{x\in\Z,\ |x|\le n \text{ such that } \tilde Q W_n(x)\le n^{-3}\}\\
\quad +\sum_{|x|\le n}\ind_{\{\tilde Q W_n(x)> n^{-3}\}}  \tilde Q W_n(x) n^{3(1-\theta)}.
\end{multline}
If $n$ is sufficiently large ( i.e., $\gb$ sufficiently small) the first term on the right-hand side is smaller than $1/n$ so that

\begin{equation}
 \sum_{|x|\le n} \left(\tilde Q W_n(x)\right)^{\theta}\le \exp(3)\tilde Q W_n+\frac{1}{n}.
\end{equation}
We are left with showing that the expectation of $W_n$ with respect to the measure $\tilde Q$ is small.
It follows from the definition of $\tilde Q$ that

\begin{equation}
 \tilde Q W_n= P\exp\left(-\gb\delta_n\#\{i\ |\ (i,S_i)\in J_n\}\right),
\end{equation}
and therefore

\begin{equation}
 \tilde Q W_n\le P\{\text{the trajectory $S$ goes out of $J_n$}\}+\exp(-n\gb\gd_n).
\end{equation}

One can choose $C_2$ such that the first term is small, and the second term is equal to
$\exp(-\gb n^{1/4}/\sqrt{2C_2\log n})\le \exp(-C_1^{1/4}/4\sqrt{C_2})$ that can be arbitrarily small by
choosing $C_1$ large compared to $(C_2)^{1/2}$.
In that case \eqref{eq:aprouver} is satisfied and we have
\begin{equation}
 p(\gb)\le \frac{1}{\theta n}{\log e^{-1}}\le -\frac{\gb^4}{2 C_1|\log \gb|^2}
\end{equation}
for small enough $\gb$. 
\end{proof}

\section{Proof of the upper bound of Theorem \ref{pasgaussi} and \ref{gaussi}}\label{onedim}

The upper bound we found in the previous section is not optimal, and can be improved by replacing the finite volume criterion \eqref{eq:aprouver} by a more sophisticated coarse graining method. The technical advantage of the coarse graining we use, is that we will not have to choose the $\theta$ of the fractional moment close to $1$ as we did in the previous section and this is the way we get rid of the extra $\log$ factor we had.
The idea of using this type of coarse graining for the copolymer model appeared in \cite{F} and this has been a substantial source of inspiration for this proof.

We will prove the following result first in the case of Gaussian environment, and then adapt the proof to general environment.


\begin{proof}[Proof in the case of Gaussian environment]

Let $n$ be the smallest squared integer bigger than $C_3\gb^{-4}$ (if $\gb$ is small we are sure that $n\le 2C_3\gb^{-4}$).  The number $n$ will be used in the sequel of the proof as a scaling factor. Let $\theta<1$ be fixed (say $\theta=1/2$). We consider a system of size $N=n m$ (where $m$ is meant to tend to infinity).

Let $I_k$ denote the interval $I_k=[k\sqrt{n},(k+1)\sqrt n )$. In order to estimate $Q W_N^{\theta}$ we decompose $W_N$ according to the contribution of different families path:

\begin{equation}\label{eq:decompo}
W_{N}=\sum_{y_1, y_2,\dots, y_m\in \Z} \check{W}_{(y_1,y_2,\dots,y_m)}
\end{equation}
where 

\begin{equation}
\check{W}_{(y_1,y_2,\dots,y_m)}=P \exp \left[\sum_{i=1}^N \left(\gb\go_{i,S_i}-\frac{\gb^2}{2}\right)\ind_{\left\{S_{in}\in I_{y_i},\forall i= 1,\dots,m \right\}}\right].
\end{equation}

\begin{figure}[h]
\begin{center}
\leavevmode
\epsfysize =6.5 cm
\psfragscanon
\psfrag{O}[c]{\tiny{O}}
\psfrag{n}[c]{\tiny{$n$}}
\psfrag{2n}[c]{\tiny{$2n$}}
\psfrag{3n}[c]{\tiny{$3n$}}
\psfrag{4n}[c]{\tiny{$4n$}}
\psfrag{5n}[c]{\tiny{$5n$}}
\psfrag{6n}[c]{\tiny{$6n$}}
\psfrag{7n}[c]{\tiny{$7n$}}
\psfrag{8n}[c]{\tiny{$8n$}}
\psfrag{rn}[c]{\tiny{$+\sqrt{n}$}}
\psfrag{r2n}[c]{\tiny{$+2\sqrt{n}$}}
\psfrag{r3n}[c]{\tiny{$+3\sqrt{n}$}}
\psfrag{r4n}[c]{\tiny{$+4\sqrt{n}$}}
\psfrag{m1n}[c]{\tiny{$-\sqrt{n}$}}
\psfrag{m2n}[c]{\tiny{$-2\sqrt n$}}
\psfrag{m3n}[c]{\tiny{$-3\sqrt n$}}
\psfrag{1}[c]{\tiny{$1$}}
\psfrag{2}[c]{\tiny{$2$}}
\psfrag{3}[c]{\tiny{$3$}}
\psfrag{4}[c]{\tiny{$4$}}
\psfrag{5}[c]{\tiny{$5$}}
\psfrag{6}[c]{\tiny{$6$}}
\psfrag{7}[c]{\tiny{$7$}}
\psfrag{8}[c]{\tiny{$8$}}
\epsfbox{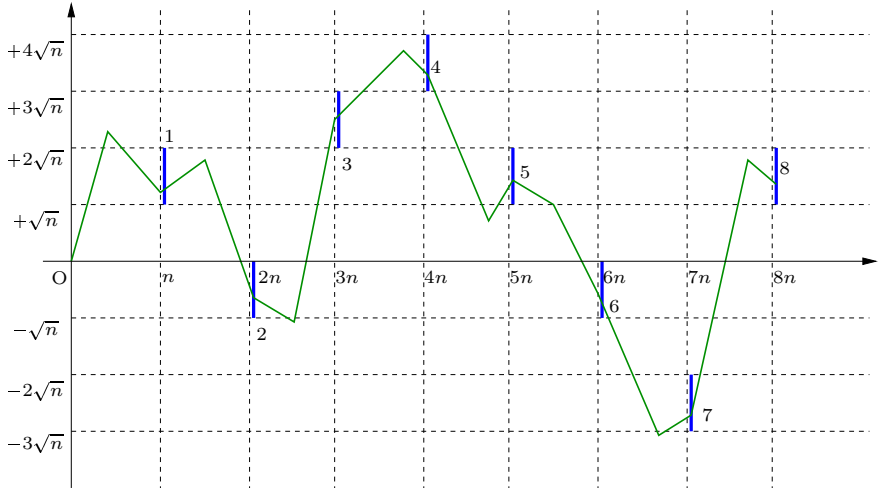}
\end{center}
\caption{\label{fig:wnn} The partition of $W_{nm}$ into $\check W^{(y_1,\dots,y_m)}$ is to be viewed as a coarse graining. For $m=8$, $(y_1,\dots,y_8)=(1,-1,2,3,1,-1,-3,1)$, $\check W_n^{(y_1,\dots,y_m)}$ corresponds to the contribution to $W_N$ of the path going through the thick barriers on the figure.}
\end{figure}

Then, we apply the inequality $\left(\sum a_i\right)^\theta\le \sum a_i^{\theta}$ (which holds for any finite or countable collection of positive real numbers) to this decomposition and average with respect to $Q$ to get,

\begin{equation}\label{eq:decompo2}
Q W_{nm}^{\theta}\le \sum_{y_1, y_2,\dots, y_m\in \Z} Q \check{W}_{(y_1,y_2,\dots,y_m)}^{\theta}.
\end{equation}
In order to estimate $Q\check{W}_{(y_1,y_2,\dots,y_m)}^{\theta}$, we use an auxiliary measure as in the previous section. The additional idea is to make the measure change depend on $y_1,\dots,y_m$.

For every $Y=(y_1,\dots,y_m)$ we define the set $J_Y$ as 
\begin{equation}
 J_Y:=\left\{ (km+i,y_k\sqrt{n}+z),\ k=0,\dots,m-1,\  i=1,\dots,n,\  |z|\le C_4\sqrt{n}\right\},
\end{equation}
where $y_0$ is equal to zero. Note that for big values of $n$ and $m$

\begin{equation} \label{cardj}
 \# J_Y\sim 2 C_4 m n^{3/2}
\end{equation}
We define the measure $\tilde Q_Y$ the measure under which the $\go_{(i,x)}$ are independent Gaussian variables with variance $1$ and mean $\tilde Q_Y \go_{(i,x)}=-\gd_n\ind_{\{(i,x)\in J_Y\}}$ where $\gd_n=n^{-3/4}C_4^{-1/2}$.
\begin{figure}[h]
\begin{center}
\leavevmode
\epsfysize =6.5 cm
\psfragscanon
\psfrag{O}[c]{\tiny{O}}
\psfrag{n}[c]{\tiny{$n$}}
\psfrag{2n}[c]{\tiny{$2n$}}
\psfrag{3n}[c]{\tiny{$3n$}}
\psfrag{4n}[c]{\tiny{$4n$}}
\psfrag{5n}[c]{\tiny{$5n$}}
\psfrag{6n}[c]{\tiny{$6n$}}
\psfrag{7n}[c]{\tiny{$7n$}}
\psfrag{8n}[c]{\tiny{$8n$}}
\psfrag{rn}[c]{\tiny{$+\sqrt{n}$}}
\psfrag{r2n}[c]{\tiny{$+2\sqrt{n}$}}
\psfrag{r3n}[c]{\tiny{$+3\sqrt{n}$}}
\psfrag{r4n}[c]{\tiny{$+4\sqrt{n}$}}
\psfrag{m1n}[c]{\tiny{$-\sqrt{n}$}}
\psfrag{m2n}[c]{\tiny{$-2\sqrt n$}}
\psfrag{m3n}[c]{\tiny{$-3\sqrt n$}}
\psfrag{1}[c]{\tiny{$1$}}
\psfrag{2}[c]{\tiny{$2$}}
\psfrag{3}[c]{\tiny{$3$}}
\psfrag{4}[c]{\tiny{$4$}}
\psfrag{5}[c]{\tiny{$5$}}
\psfrag{6}[c]{\tiny{$6$}}
\psfrag{7}[c]{\tiny{$7$}}
\psfrag{8}[c]{\tiny{$8$}}
\psfrag{ZMO}[c]{\small{Region where the environment is modified}}
\epsfbox{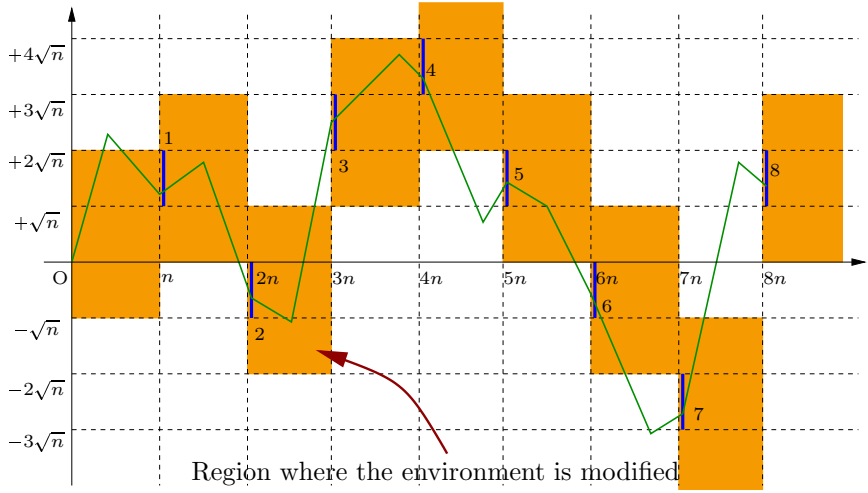}
\end{center}
\caption{\label{fig:wnn2} This figure represent in a rough way the change of measure $Q_Y$. The region where the mean of $\go_{(i,x)}$ is lowered (the shadow region on the figure) corresponds to the region where the simple random walk is likely to go, given that it goes through the thick barriers.}
\end{figure}
The law $\tilde Q_Y$ is absolutely continuous with respect to $Q$ and its density is equal to

\begin{equation}\label{qy1}
\frac{\dd\tilde Q_Y}{\dd Q}(\go)=\exp\left(-\sum_{(i,x)\in J_Y}\left[\gd_n\go_{(i,x)}+\gd_n^2/2\right]\right).
\end{equation}
Using H\"older inequality with this measure as we did in the previous section, we obtain

\begin{multline}
 Q \left[\check W_{(y_1,y_2,\dots,y_m)}^{\theta}\right]=\tilde Q_{Y} \left[\frac{\dd Q}{\dd \tilde Q_{Y}}\check W_{(y_1,y_2,\dots,y_m)}^{\theta}\right]\\
 \le \tilde Q_{Y}\left( \left[\left(\frac{\dd Q}{\dd \tilde Q_{Y}}\right)^{\frac{1}{1-\theta}}\right]\right)^{1-\theta}\left( \tilde Q_{Y} \check W_{(y_1,\dots,y_m)}\right)^{\theta} \label{eq:hhol}.
\end{multline}
The value of the first term can be computed explicitly

\begin{equation}
\left(Q \left[\left(\frac{\dd Q}{\dd \tilde Q_{Y}}\right)^{\frac{\theta}{1-\theta}}\right]\right)^{1-\theta}=\exp\left(\frac{\# J_Y\theta\gd_n^2}{2(1-\theta)}\right)\le \exp(3m), \label{eq:dens}
\end{equation}
where the upper bound is obtained by using the definition of $\gd_n$, \eqref{cardj} and the fact that $\theta=1/2$.

Now we compute the second term

\begin{equation}\label{ccek}
\tilde Q_{Y} \check W_{(y_1,\dots,y_m)}=P \exp\left(-\gb\gd_n\#\left\{i|(i,S_i)\in J_Y\right\}\right)\ind_{\{S_{kn}\in I_{y_k},\ \forall k \in [1,m]\}}.
\end{equation}
We define 
\begin{equation}\label{jayjay}\begin{split}
 J&:= \{(i,x),\ i=1,\dots,n,\ |x|\le C_4\sqrt{n}\}\\
 \bar J&:= \{(i,x),\ i=1,\dots,n,\ |x|\le (C_4-1)\sqrt{n}\}.
\end{split}
\end{equation}
Equation \eqref{ccek} implies that (recall that $P_x$ is the law of the simple random walk starting from $x$, and that we set $y_0=0)$
\begin{equation}\label{maxstart}
 \tilde Q_{Y} \check W_{(y_1,\dots,y_m)}\le \prod_{k=1}^m \max_{x\in I_0} P_x \exp\left(-\gb\gd_n\#\left\{i\ :\ (i,S_i)\in J\right\}\right)\ind_{\{S_{n}\in I_{y_k-y_{k-1}}\}}.
\end{equation}
Combining this with \eqref{eq:decompo}, \eqref{eq:hhol} and \eqref{eq:dens} we have

\begin{equation}
\log Q W_{N}^{\theta}\le m\left[3+\log \sum_{y\in\Z}\left(\max_{x\in I_0} P_x \exp\left(-\gb\gd_n\#\left\{i\ :\ (i,S_i)\in J\right\}\right)\ind_{\{S_{n}\in I_{y}\}}\right)^{\theta}\right].
\end{equation}
If the quantity in the square brackets is smaller than $-1$, by equation \eqref{momentfrac} we have  $p(\gb)\le -1/n$.
Therefore, to complete the proof it is sufficient to show that 
\begin{equation}
\sum_{y\in\Z}\left(\max_{x\in I_0} P_x \exp\left(-\gb\gd_n\#\left\{i\ :\ (i,S_i)\in J\right\}\right)\ind_{\{S_{n}\in I_{y}\}}\right)^{\theta} 
\end{equation}
is small. To reduce the problem to the study of a finite sum, we observe  (using some well known result on the asymptotic behavior of random walk) that given $\gep>0$ we can find $R$ such that

\begin{equation}\label{iop}
 \sum_{|y|\ge R}\left(\max_{x\in I_0} P_x \exp\left(-\gb\gd_n\#\left\{i\ :\ (i,S_i)\in J\right\}\right)\ind_{\{S_{n}\in I_{y}\}}\right)^{\theta}\le \sum_{|y|\ge R}\max_{x\in I_0} \left(P_x\{S_{n}\in I_{y}\}\right)^{\theta}\le \gep.
\end{equation}
To estimate the remainder of the sum we use the following trivial bound

\begin{multline}
 \sum_{|y|< R}\left(\max_{x\in I_0} P_x \exp\left(-\gb\gd_n\#\left\{i\ :\ (i,S_i)\in J\right\}\right)\ind_{\{S_{n}\in I_{y}\}}\right)^{\theta}\\
\le 2 R \left(\max_{x\in I_0} P_x \exp\left(-\gb\gd_n\#\left\{i\ :\ (i,S_i)\in J\right\}\right)\right)^{\theta}.
\end{multline}
Then we get rid of the $\max$ in the sum by observing that if a walk starting from $x$ makes a step in $J$, the walk with the same increments starting from $0$ will make the same step in $\bar J$ (recall \eqref{jayjay}).
\begin{equation} \label{fgh}
\max_{x\in I_0} P_x \exp\left(-\gb\gd_n\#\left\{i\ :\ (i,S_i)\in J\right\}\right)\le P \exp\left(-\gb\gd_n\#\left\{i|(i,S_i)\in \bar J\right\}\right).
\end{equation}
Now we are left with something similar to what we encountered in the previous section

\begin{equation} \label{jkl}
P \exp\left(-\gb\gd_n\#\left\{i\ :\ (i,S_i)\in \bar J\right\}\right)\le P\{\text{ the random walk goes out of $\bar J$ }\} + \exp(-n \gb \gd_n). 
\end{equation}
If $C_4$ is chosen large enough, the first term can be made arbitrarily small by choosing $C_4$ large, and the second is equal to $\exp(-C_3^{-1/4}/\sqrt{C_4})$ and can be made also arbitrarily small if $C_3$ is chosen large enough once $C_4$ is fixed.
An appropriate choice of constant and the use of \eqref{fgh} and \eqref{jkl} can  leads then to

\begin{equation} \label{klm}
 2 R \left(\max_{x\in I_0} P_x \exp\left(-\gb\gd_n\#\left\{i\ :\ (i,S_i)\in J\right\}\right)\right)^{\theta}\le \gep.
\end{equation}
This combined with \eqref{iop} completes the proof.
\end{proof}
\bigskip




\begin{proof}[Proof of the general case]
In the case of a general environment, some modifications have to be made in the proof above, but the general idea remains the same. In the change of measure one has to change the shift of the environment in $J_Y$ \eqref{qy1} by an exponential tilt of the measure as follow

\begin{equation}
\frac{\dd\tilde Q_{Y}}{\dd Q}(\gb)=\exp\left(-\sum_{(i,z)\in J_Y}\left[\gd_n\go_{(i,z)}+\gl(-\gd_n) \right]\right).
\end{equation}
The formula estimating the cost of the change of measure \eqref{eq:dens} becomes

\begin{equation}
\left(Q \left(\frac{\dd Q}{\dd \tilde Q_{Y}}\right)^{\frac{\theta}{1-\theta}}\right)^{1-\theta}=\exp\left(\# J_Y \left[(1-\theta)\gl\left(\frac{\theta\gd_n}{1-\theta}\right)+\theta\gl(-\gd_n)\right]\right)\le \exp(2 m),
\end{equation}
where the last inequality is true if $\gb_n$ is small enough if we consider that $\theta=1/2$ and use the fact
that $\gl(x)\stackrel{x\to 0}{\sim}x^2/2$ ($\go$ has $0$ mean and unit variance). The next thing we have to do is to compute the effect of this change of measure in this general case, i.e.\ find an equivalent for \eqref{ccek}.
When computing $\tilde Q_Y \check W_{(y_1,\dots,y_m)}$, the quantity 
\begin{equation}
\tilde Q_Y \exp(\gb\go_{1,0}-\gl(\gb))=\exp\left[\gl(\gb-\gd_n)-\gl(-\gd_n)-\gl(\gb)\right]
\end{equation}
appears instead of $\exp(-\gb\gd_n)$.
Using twice the mean value theorem, one gets that there exists $h$ and $h'$ in $(0,1)$ such that

\begin{equation}
 \gl(\gb-\gd_n)-\gl(-\gd_n)-\gl(\gb)=\gd_n\left[\gl'(-h\gd_n)-\gl'(\gb-h\gd_n)\right]=-\gb\gd_n\gl''(-h\gd_n+h'\gb).
\end{equation}
And as $\go$ has unit variance $\lim_{x\to 0} \gl''(x)=1$. Therefore if $\gb$ and $\gd_n$ are chosen small enough, the right-hand side of the above is less than $-\gb\gd_n/2$. So that \eqref{ccek} can be replaced by

\begin{equation}
\tilde Q_{Y} \check W_{(y_1,\dots,y_m)}\le P \exp\left(-\frac{\gb\gd_n}{2}\#\left\{i|(i,S_i)\in J_Y\right\}\right)\ind_{\{S_{kn}\in I_{y_k},\ \forall k \in [1,m]\}}.
\end{equation}

The remaining steps follow closely the argument exposed for the Gaussian case.
\end{proof}

\section{Proof of the upper bound in Theorem \ref{1+2uplb}}\label{twodim}

In this section, we prove the main result of the paper: very strong disorder holds at all temperature in dimension $2$. 


\medskip

The proof is technically quite involved. It combines the tools of the two previous sections with a new idea for the change a measure: changing the covariance structure of the environment. We mention that this idea was introduced recently in \cite{GLTm} to deal with the {\sl marginal disorder} case in pinning model. We choose to present first a proof for the Gaussian case, where the idea of the change a measure is easier to grasp.

\medskip

Before starting, we sketch the proof and how it should be decomposed in different steps:
\begin{itemize}
 \item [(a)] We reduce the problem by showing that it is sufficient to show that
 for some real number $\theta<1$, $Q W_N^{\theta}$ decays exponentially with $N$.
 \item [(b)] We use a coarse graining decomposition of the partition function by splitting it into different contributions that corresponds to trajectories that stays in a large corridor. This decomposition is similar to the one used in Section \ref{onedim}.
 \item [(c)] To estimate the fractional moment terms appearing in the decomposition, we change the law of the environment around the corridors corresponding to each contribution. More precisely, we introduce negative correlations into the Gaussian field of the environment. We do this change of measure in such a way that 
the new measure is not very different from the original one.
 \item [(d)] We use some basic properties of the random walk in $\Z^2$ to compute the expectation under the new measure.
\end{itemize}

\medskip

\begin{proof}[Proof for Gaussian environment]
 We fix $n$ to be the smallest squared integer bigger than $\exp(C_5/\gb^4)$ for some large constant $C_5$ to be defined later, for small $\gb$ we have $n\le \exp(2 C_5/\gb^4)$. The number $n$ will be used in the sequel of the proof as a scaling factor. 
For $y=(a,b)\in \Z^2$ we define 
$I_y=[a\sqrt{n},(a+1)\sqrt n-1]\times[b\sqrt{n},(b+1)\sqrt{n}-1]$ so that $I_y$ are disjoint and cover $\Z^2$.
For $N=nm$, we decompose the normalized partition function $W_N$ into different contributions, very similarly to what is done in dimension one (i.e.\ decomposition \eqref{eq:decompo2}), and we refer to the figure \ref{fig:wnn2} to illustrate how the decomposition looks like:

\begin{equation}
W_N=\sum_{y_1,\dots,y_m\in \Z^2}\check W_{(y_1,\dots,y_m)}
\end{equation}
where
\begin{equation}
\check{W}_{(y_1,\dots,y_m)}=P \exp \left(\sum_{i=1}^N \left[\gb\go_{i,S_i}-\gb^2/2\right] \right)\ind_{\left\{S_{in}\in I_{y_i},\forall i= 1,\dots,m \right\}}.
\end{equation}
We fix $\theta<1$ and apply the inequality $(\sum a_i)^{\theta}\le \sum a_i^{\theta}$ (which holds for any finite or countable collection of positive real numbers) to get

\begin{equation}\label{bobbob}
 Q W_N^{\theta}\le \sum_{y_1,\dots,y_m\in \Z^2}Q\check W_{(y_1,\dots,y_m)}^{\theta}.
\end{equation}

In order to estimate the different terms in the sum of the right--hand side in \eqref{bobbob}, we define some auxiliary measures $\tilde Q_Y$ on the the environment for every $Y=(y_0,y_1,\dots,y_m)\in \Z^{d+1}$ with $y_0=0$. We will choose the measures $Q_Y$ absolutely continuous with respect to $Q$. We use H\"older inequality to get the following upper bound:

\begin{equation}\label{hhold}
Q \check W_{(y_1,\dots,y_m)}^{\theta}\le  \left(Q \left(\frac{\dd Q}{\dd\tilde Q_{Y}}\right)^{\frac{\theta}{1-\theta}}\right)^{1-\theta} \left( \tilde Q_{Y} \check W_{(y_1,\dots,y_m)}\right)^{\theta}.
\end{equation}

Now, we describe the change of measure we will use. Recall that for the $1$-dimensional case we used a shift of the environment along the corridor corresponding to $Y$. The reader can check that this method would not give the exponential decay of $W_N$ in this case. Instead we change the covariance function of the environment along the corridor on which the walk is likely to go by introducing some negative correlation.

We introduce the change of measure that we use for this case. Given $Y=(y_0,y_1,\dots,y_m)$ we define $m$ blocks $(B_k)_{k\in[1,m]}$ and $J_Y$ their union (here and in the sequel, $|z|$ denotes the $l^{\infty}$ norm on $\Z^2$):

\begin{equation}\begin{split}\label{blocdef}
 B_k&:=\left\{(i,z)\in \N\times \Z^2 \, : \, \lceil i/n \rceil =k \text{ and } |z-\sqrt{n} y_{k-1}|\le C_6 \sqrt{n}\right\},\\
 J_Y&:=\bigcup_{k=1}^m B_k.
\end{split}
\end{equation}
We fix the covariance the field $\go$ under the law $\tilde Q_{Y}$ to be equal to

\begin{equation}\begin{split}\label{matvar}
&\tilde Q_Y \left(\go_{i,z}\go_{i,z'}\right)=\mathcal C^Y_{(i,z),(j,z')}\\
 &:=\begin{cases}
   \ind_{\{(i,z)=(j,z')\}}-V_{(i,z),(j,z')} & \text{ if } \exists\ k\in[1,m] \text{ such that } (i,z) \text{ and } (j,z')\in B_k
   \\
    \ind_{\{(i,z)=(j,z')\}} & \text{ otherwise,}
  \end{cases}\end{split}
\end{equation}
where

\begin{equation}\label{vdef}
 V_{(i,z),(j,z')}:=\begin{cases}
  0 & \text{ if } (i,z)=(j,z')
   \\
 \frac{\ind_{\{|z-z'|\le C_7 \sqrt{|j-i|}\}}}{100 C_6C_7n\sqrt{\log n}|j-i|} & \text{ otherwise.}
\end{cases}
\end{equation}
We define

\begin{equation}
\hat V:=(V_{(i,z),(j,z')})_{(i,z),(j,z')\in B_1}.
\end{equation}

One remarks that the so-defined covariance matrix  $\mathcal C^Y$ is block diagonal with $m$ identical blocks which are copies of $I-\hat V$ corresponding to the $B_k$, $k\in [1,m]$, and just ones on the diagonal elsewhere.
Therefore, the change of measure we describe here exists if and only if  $I-\hat V$ is definite positive. 

The largest eigenvalue for $\hat V$ is associated to a  positive vector and therefore is smaller than 
\begin{equation}\label{specray}
 \max_{(i,z)\in B_1}\sum_{(j,z')\in{B_1}} \left|V_{(i,z),(j,z')}\right|\le \frac{C_7}{C_6 \sqrt{\log n}}.
\end{equation}
For the sequel we choose $n$ such that the spectral radius of $\hat V$ is less than $(1-\theta)/2$ so 
that $I-\hat V$ is positive definite. 
With this setup, $\tilde Q_Y$ is well defined.

The density of the modified measure $\tilde Q_{Y}$ with respect to $Q$ is
given by
\begin{equation}
 \frac{\dd \tilde Q_Y}{\dd Q}(\go)=\frac{1}{\sqrt{\det \mathcal C^Y}}\exp\left(-\frac{1}{2}^t\go((\mathcal C^Y)^{-1}-I)\go\right),
\end{equation}
where 
\begin{equation}
^t\go M \go= \sum_{(i,z),(j,z')\in \N\times \Z^2}\go_{(i,z)}M_{(i,z),(j,z')}\go_{(j,z')},
\end{equation}
for any matrix $M$ of $(\N\times \Z^2)^2$ with finite support.

Then we can compute explicitly the value of the second term in the right-hand side of \eqref{hhold}
\begin{equation}\label{detei}
\left(Q\left(\frac{\dd Q}{\dd\tilde Q_Y}\right)^{\frac{\theta}{1-\theta}}\right)^{1-\theta}=\sqrt{\frac{\det \mathcal C^Y}{\det \left(\frac{\mathcal C^Y}{1-\theta}-\frac{\theta I}{1-\theta}\right)^{1-\theta}}}.
\end{equation}
Note that the above computation is right if and only if $\mathcal C^{Y} -\theta I$ is a definite positive matrix.
Since its eigenvalues are the same of those of $(1-\theta)I-\hat{V}$, this holds for large $n$ thanks to \eqref{specray}. 
Using again the fact that $\mathcal C^Y$ is composed of $m$ blocks identical to $I-\hat V$, we get from \eqref{detei}
\begin{equation}\label{detbl}
\left(Q\left(\frac{\dd Q}{\dd\tilde Q}\right)^{\frac{\theta}{1-\theta}}\right)^{1-\theta}=\left(\frac{\det(I
-\hat V)}{\det(I-\hat V/(1-\theta))^{1-\theta}}\right)^{m/2}.
\end{equation}
In order to estimate the determinant in the denominator, we compute the Hilbert-Schmidt norm of $\hat V$. One can check that for all $n$

\begin{equation} \label{hsnorm}
 \| \hat V \|^2= \sum_{(i,z),(j,z')\in B_1} V_{(i,z),(j,z')}^2\le 1.
\end{equation}
We use the inequality $\log( 1+x )\ge x-x^2$ for all $x\ge -1/2$ and the fact that the spectral radius of $\hat V/(1-\theta)$ is bounded by $1/2$ (cf. \eqref{specray}) to get that

\begin{equation}\label{truv}\begin{split}
 \det\left[I-\frac {\hat V}{1-\theta}\right]&=\exp\left(\tr\left(\log\left(I-\frac{\hat V}{1-\theta}\right)\right)\right)\ge \exp\left(-\frac{\| \hat V\|^2}{(1-\theta)^2}\right)\\
&\ge \exp\left(-\frac{1}{(1-\theta)^2}\right).
\end{split}\end{equation}
For the numerator, $\tr\ \hat V=0$ implies that that $\det(I-\hat V)\le 1$.
Combining this with \eqref{detbl} and \eqref{truv} we get

\begin{equation}\label{densityy}
\left(Q\left(\frac{\dd Q}{\dd\tilde Q_Y}\right)^{\frac{\theta}{1-\theta}}\right)^{1-\theta}\le \exp\left(\frac{m}{2(1-\theta)}\right).
\end{equation}
Now that we have computed the term corresponding to the change of measure, we estimate $\check W_{(y_1,\dots,y_m)}$ under the modified measure (just by computing the variance of the Gaussian variables in the exponential, using \eqref{matvar})\ :

\begin{multline}\label{modest}
 \tilde Q_{Y} \check W_{(y_1,\dots,y_m)}=P\, \tilde Q_Y \exp\left( \sum_{i=1}^N \left(\gb \go_{i,S_i}-\frac{\gb^2}{2}\right)\right)\ind_{\left\{S_{kn}\in I_{y_k},\forall k= 1,\dots,m \right\}}\\
=P \exp\left(\frac{\gb^2}{2}\sumtwo{1\le i,\ j\le N}{z,z'\in \Z^2} \left(\mathcal C^Y_{(i,z),(j,z')}-\ind_{\{(i,z)=(j,z')\}} \right)\ind_{\{S_i=z,S_j=z'\}}\right)\ind_{\left\{S_{kn}\in I_{y_k},\forall k= 1,\dots,m \right\}}.
\end{multline}
Replacing $\mathcal C^Y$ by its value we get that

\begin{multline}
 \tilde Q_{Y} \check W_{(y_1,\dots,y_m)}=P \exp  \left(-\frac{\gb^2}{2}\sumtwo{1\le i\neq j\le N}{1\le k\le m}\frac{\ind_{\{\left((i,S_i),(j,S_j)\right)\in B_k^2,\ |S_i-S_j|\le C_7\sqrt{|i-j|}\}}}{100C_6C_7 n \sqrt{\log n}|j-i|}\right)
\\
\ind_{\left\{S_{kn}\in I_{y_k},\forall k= 1,\dots,m \right\}}.
\end{multline}
Now we do something similar to $\eqref{maxstart}$: for each ``slice'' of the trajectory $(S_i)_{i\in[(m-1)k,mk]}$, we bound the contribution of the above expectation by maximizing over the starting point (recall that $P_x$ denotes the probability distribution of a random walk starting at $x$). Thanks to the conditioning, the starting point has to be in $I_{y_k}$. Using the translation invariance of the random walk, this gives us the following ($\vee$ stands for maximum):

\begin{multline} \label{rrrr} 
 \tilde Q_{Y} \check W^{(y_1,\dots,y_m)}\le \prod_{i=k}^m \max_{x\in I_0}\ P_x\bigg[\\
\exp\left(-\frac{\gb^2}{2}\sum_{1\le i\neq j\le n}\frac{\ind_{\{|S_i|\vee|S_j|\le C_6\sqrt{n}, \ |S_i-S_j|\le C_7\sqrt{|i-j|}\}}}{100 C_6 C_7 n \sqrt{\log n} |j-i|}\right)
\ind_{\left\{S_{n}\in I_{y_k-y_{k-1}}\right\}}\bigg].
\end{multline}
For trajectories $S$ of a directed random-walk of $n$ steps, we define the quantity 

\begin{equation}
G(S):=\sum_{1\le i\neq j\le n}\frac{\ind_{\{|S_i|\vee |S_j|\le C_6\sqrt{n}, \ |S_i-S_j|\le C_7\sqrt{|i-j|}\}}}{100 C_6 C_7 n \sqrt{\log n} |j-i|}.
\end{equation}
Combining \eqref{rrrr} with \eqref{densityy}, \eqref{hhold} and \eqref{bobbob}, we finally get

\begin{equation}\label{aaaaargh}
Q W_N^{\theta}\le \exp\left(\frac{m}{2(1-\theta)}\right)\left[\sum_{y\in \Z}\max_{x\in I_0} \left(P_x \exp\left(-\frac{\gb^2}{2}G(S)\right)\ind_{\{S_n\in I_{y}\}}\right)^{\theta}\right]^m.
\end{equation}
 The exponential decay of $Q W_N^{\theta}$ (with rate $n$) is guaranteed if we can prove that

\begin{equation}
 \sum_{y\in \Z}\max_{x\in I_0}\left(P_x \exp\left(-\frac{\gb^2}{2}G(S)\right)\ind_{\{S_n\in I_{y}\}}\right)^{\theta}
\end{equation}
is small. The rest of the proof is devoted to that aim.

We fix some $\gep>0$. Asymptotic properties of the simple random walk, guarantees that we can find $R=R_{\gep}$ such that
\begin{equation}\label{tr0}
  \sum_{|y|\ge R}\max_{x\in I_0}\left(P_x \exp\left(-\frac{\gb^2}{2}G(S)\right)\ind_{\{S_n\in I_{y}\}}\right)^{\theta}
\le \sum_{|y|\ge R} \max_{x\in I_0}\ \left(P_x\{S_n\in I_y\}\right)^{\theta}\le \gep.
\end{equation}
To estimate the rest of the sum, we use the following trivial and rough bound

\begin{equation}\label{tr10}
 \sum_{|y|< R}\max_{x\in I_0}\left[ P_x \exp\left(-\frac{\gb^2}{2}G(S)\right)\ind_{\{S_n\in I_{y}\}}\right]^{\theta}
\le R^2 \left[\max_{x\in I_0} P_x \exp\left(-\frac{\gb^2}{2}G(S)\right)\right]^{\theta}
\end{equation}
Then we use the definition of $G(S)$ to get rid of the $\max$ by reducing the width of the zone where we have negative correlation:

\begin{equation}
 \max_{x\in I_0} P_x \exp\left(-\frac{\gb^2}{2}G(S)\right)\le P \exp\left(-\frac{\gb^2}{2}\sum_{1\le i\neq j\le n}\frac{\ind_{\{|S_i|\vee |S_j|\le (C_6-1)\sqrt{n}, \ |S_i-S_j|\le C_7\sqrt{|i-j|}\}}}{100 C_6 C_7 n \sqrt{\log n} |j-i|}\right).
\end{equation}
We define $\bar B:=\{(i,z)\in \N\times \Z^2\ : \ i\le m, |z|\le (C_6-1)\sqrt{n}\}$. We get from the above that

\begin{multline}\label{tr11}
  \max_{x\in I_0} P_x \exp\left(-\frac{\gb^2}{2}G(S)\right)\le  P\{\text{the RW goes out of } \bar B \}\\
+
P \exp\left(-\frac{\gb^2}{2}\sum_{1\le i\neq j\le n}\frac{\ind_{\{ |S_i-S_j|\le C_7\sqrt{|i-j|}\}}}{100 C_6 C_7 n \sqrt{\log n} |j-i|}\right)
\end{multline}
One can make the first term of the right-hand side arbitrarily small by choosing $C_6$ large, in particular on can choose $C_6$ such that

 \begin{equation}\label{tr1}
 P\left\{ \max_{i\in[0,n]} |S_n|\ge (C_6-1)\sqrt{n}\right\}\le (\gep/R^2)^{\frac{1}{\theta}}.
\end{equation}
To bound the other term, we introduce the quantity

\begin{equation}
D(n):=\sum_{1\le i\neq j\le n}\frac{1}{n \sqrt{\log n} |j-i|},
\end{equation}
and the random variable $X$,

\begin{equation}\label{xdef}
 X:=\sum_{1\le i\neq j\le n}\frac{\ind_{\{|S_i-S_j|\le C_7\sqrt{|i-j|}\}}}{n \sqrt{\log n} |j-i|}.
\end{equation}
For any $\delta>0$, we can find $C_7$ such that $P(X)\ge( 1-\gd )D(n)$.
We fix $C_7$ such that this holds for some good $\gd$ (to be chosen soon), and by remarking that $0\le X \le D(n)$ almost surely, we obtain
(using Markov inequality)

\begin{equation}\label{xest}
 P\{ X>D(n)/2 \}\ge 1-2\delta.
\end{equation}
Moreover we can estimate $D(n)$ getting that for $n$ large enough

\begin{equation}\label{dest}
 D(n)\ge \sqrt{\log n}.
\end{equation}
Using \eqref{xest} and \eqref{dest} we get

\begin{equation}\begin{split}
 P \exp\left(-\frac{\gb^2}{2}\sum_{1\le i\neq j\le n}\frac{\ind_{\{ |S_i-S_j|\le C_7\sqrt{|i-j|}\}}}{100C_6 C_7 n \sqrt{\log n} |j-i|}\right)&=P \exp\left(-\frac{\gb^2}{200C_6C_7}X\right)\\
&\le 2\gd + \exp\left(-\frac{\gb^2\sqrt{\log n}}{200C_6C_7}\right).
\end{split}\end{equation}
Due to the choice of $n$ we have made (recall $n\ge \exp(C_5/\gb^4)$), the second term is less than
$\exp\left(-\gb^2 C_5^{1/2}/(200 C_6C_7)\right)$. We can choose $\gd$, $C_7$ and $C_5$ such that, the right-hand side is less that $(\gep/R^2)^{\frac{1}{\theta}}$.
This combined with \eqref{tr1}, \eqref{tr11}, \eqref{tr10} and \eqref{tr0} allow us to conclude that

\begin{equation}
 \sum_{y\in \Z}\max_{x\in I_0}\left(P_x \exp\left(-\frac{\gb^2}{2}G(S)\right)\ind_{\{S_n\in I_{y}\}}\right)^{\theta}\le 3\gep
\end{equation}
So that with a right choice for $\gep$, \eqref{aaaaargh} implies

\begin{equation}
Q W_N^{\theta}\le \exp(-m).
\end{equation}
Then \eqref{momentfrac} allows us to conclude that $p(\gb)\le -1/n$.

\end{proof}

\bigskip

\begin{proof}[Proof for general environment]
The case of general environment does not differ very much from the Gaussian case, but one has to a different approach for the change of measure in \eqref{hhold}. In this proof, we will largely refer to what has been done in the Gaussian case, whose proof should be read first.
\medskip

Let $K$ be a large constant.
One defines the function $f_K$ on $\R$ as to be
\begin{equation*}
 f_K(x)=-K\ind_{\{x>\exp(K^2)\}}.
\end{equation*}
Recall the definitions \eqref{blocdef} and \eqref{vdef}, and define $g_Y$ function of the environment as

\begin{equation*}
 g_Y(\go)=\exp\left(\sum_{k=1}^m f_K\left(\sum_{(i,z),(j,z')\in B_k} V_{(i,z),(j,z')}\go_{i,z}\go_{j,z'}\right)\right).
\end{equation*}
Multiplying by $g_Y$ penalizes by a factor $\exp(-K)$ the environment for which there is to much correlation in one block. This is a way of producing negative correlation in the environment.
For the rest of the proof we use the notation
\begin{equation}
 U_k:=\sum_{(i,z),(j,z')\in B_k} V_{(i,z),(j,z')}\go_{i,z}\go_{j,z'}
\end{equation}

We do a computation similar to \eqref{hhold} to get

\begin{equation}\label{hhold2}
 Q\left[\check W_{(y_1,\dots,y_m)}^{\theta}\right]\le \left(Q\left[ g_Y(\go)^{-\frac{\theta}{1-\theta}}\right]\right)^{1-\theta}\left( Q\left[ g_Y(\go)\check W_{(y_1,\dots,y_n)}\right]\right)^{\theta}.
\end{equation}
The block structure of $g_Y$ allows to express the first term as a power of $m$.

\begin{equation}
Q\left[ g_Y(\go)^{-\frac{\theta}{1-\theta}}\right]=\left(Q\left[\exp\left(-\frac{\theta}{1-\theta}f_K\left(U_1\right)\right)\right]\right)^m.
\end{equation}
Equation \eqref{hsnorm} says that
\begin{equation}
 \var_Q\left(U_1\right)\le 1.
\end{equation}
So that 
\begin{equation}
 P\left\{U_1\ge \exp(K^2)\right\}\le \exp(- 2 K^2),
\end{equation}
and hence
\begin{multline}\label{fsterm}
Q\left[\exp\left(-\frac{\theta}{1-\theta}f_K\left(U_1\right)\right)\right]\\
\le 1+\exp\left(-2K^2+\frac{\theta}{1-\theta}K\right)\le 2,
\end{multline}
if $K$ is large enough.
We are left with estimating the second term
\begin{equation}\label{secterm}
  Q\left[g_Y(\go)\check W_{(y_1,\dots,y_n)}\right]=P Q g_Y(\go)\exp\left(\sum_{i=1}^{nm}[\gb \go_{i,S_i}-\gl(\gb)] \right)\ind_{\{S_{kn}\in I_{y_k},\forall k=1\dots m\}}.
\end{equation}
For a fixed trajectory of the random walk $S$, we consider $\bar Q_S$ the modified measure on the environment with density

\begin{equation}
\frac{\dd \bar Q_S}{\dd Q}:=\exp\left(\sum_{i=1}^{nm}[\gb \go_{i,S_i}-\gl(\gb)] \right).
\end{equation}
Under this measure

\begin{equation}
 \bar Q_S \go_{i,z}=\begin{cases}
  0 & \text{ if } z\neq S_i
   \\
 Q \go e^{\gb\go_{0,1}-\gl(\gb)}:=m(\gb) & \text{ if } z=S_i.
\end{cases}
\end{equation}
As $\go_{1,0}$ has zero-mean and unit variance under $Q$, \eqref{expm} implies $m(\gb)=\gb+o(\gb)$ around zero and that $\var_{\bar Q _S} \go_{i,z}\le 2$ for all $(i,z)$ if $\gb$ is small enough. Moreover $\bar Q_S$ is a product measure, i.e.\ the $\go_{i,z}$ are independent variables under $Q_S$.
With this notation \eqref{secterm} becomes

\begin{equation}
 P \bar Q_S \left[g_Y(\go)\right]\ind_{\{S_{kn}\in I_{y_k},\forall k=1,\dots,m \}}.
\end{equation}
As in the Gaussian case, one wants to bound this by a product using the block structure. Similarly to \eqref{rrrr}, we use translation invariance to get the following upper bound

\begin{equation}
 \prod_{k=1}^m \max_{x\in I_0} P_x \bar Q_S \exp\left(f_K\left(U_1\right)\right)\ind_{\{S_n\in I_{y_k-y_{k-1}}\}}.
\end{equation}
Using this in \eqref{hhold2} with the bound \eqref{fsterm} we get the inequality

\begin{equation}
Q W_N^{\theta}\le 2^{m(1-\theta)}\!\! \left(\sum_{y\in \Z^2}\!\!  \left[\max_{x\in I_0} P_x \bar Q_S \exp\left(f_K\left(U_1\right)\right)\ind_{\{S_n\in I_{y}\}}\right]^{\theta}\right)^m\!\!\!\!.
\end{equation}
Therefore to prove exponential decay of $Q W_N^{\theta}$, it is sufficient to show that 

\begin{equation}
\sum_{y\in \Z^2} \left[\max_{x\in I_0} P_x \bar Q_S \exp\left(f_K\left(U_1\right)\right)\ind_{\{S_n\in I_{y}\}}\right]^{\theta}
\end{equation}
is small. As seen in the Gaussian case ( cf.\ \eqref{tr0},\eqref{tr10}), the contribution of $y$ far from zero can be controlled and therefore it is sufficient for our purpose to check
\begin{equation}
 \max_{x\in I_0} P_x \bar Q_S \exp\left(f_K\left(U_1\right)\right)\le \gd,
\end{equation}
for some small $\gd$. Similarly to \eqref{tr11}, we force the walk to stay in the zone where the environment is modified by writing
\begin{multline}\label{foor}
\max_{i\in I_0} P_x \bar Q_S \exp\left(f_K\left(U_1\right)\right)\le 
P\{ \max_{i\in[0,n]}|S_i|\ge (C_6-1)\sqrt{n}\}\\
+\max_{x\in I_0} P_x \bar Q_S \exp\left(f_K\left(U_1\right)\right)\ind_{\{|S_n-S_0|\le (C_6-1)\sqrt{n}\}}.
\end{multline}
The first term is smaller than $\gd/6$ if $C_6$ is large enough. To control the second term, we will find an upper bound for

\begin{equation}
 P_x \bar Q_S \exp\left(f_K\left(U_1\right)\right)\ind_{\{\max_{i\in[0,n]}|S_i-S_0|\le (C_6-1)\sqrt{n}\}},
\end{equation}
which is uniform in $x\in I_0$.

What we do is the following: we show that for most trajectories $S$ the term in $f_K$ has a large mean and a small variance with respect to $Q_S$ so that $f_K(\ \dots \ )=-K$ with large $\bar Q_S$ probability. The rest will be easy to control as the term in the expectation is at most one.

The expectation of $U_1$ under $\bar Q_S$ is equal to
\begin{equation}
 m(\gb)^2\sum_{1\le i,j\le n} V_{(i,S_i),(j,S_j)}.
\end{equation}
When the walk stays in the block $B_1$ we have (using definition \eqref{xdef})
\begin{equation}\label{xvij}
 \sum_{1\le i,j\le n} V_{(i,S_i),(j,S_j)}=\frac{1}{100C_6C_7}X.
\end{equation}
The distribution of $X$ under $P_x$ is the same for all $x \in I_0$. It has been shown earlier (cf.\ \eqref{xest} and \eqref{dest}), that if $C_7$ is chosen large enough, 
\begin{equation}
 P\left\{ \frac{m(\gb)^2}{100C_6C_7} X\le \frac{\sqrt{\log n}}{200 C_6C_7}\right\}\le \frac{\gd}{6}.
\end{equation}
As $m(\gb)\ge \gb/2$ if $\gb$ is small, if $C_5$ is large enough (recall $n\ge \exp(C_5/\gb^4)$), this together with \eqref{xvij} gives.
\begin{equation}\label{del2}
 P_x\bigg\{ m(\gb)^2 \bar Q_S\left(U_1\right)\le 2\exp(K^2);
 \max_{i\in[0,n]}|S_i-S_0|\le (C_6-1)\sqrt{n}\bigg\}\le \frac{\gd}{6}.
\end{equation}
To bound the variance of $U_1$ under $\bar Q_S$, we decompose the sum

\begin{multline}
 U_1=\sum_{(i,z),(j,z')\in B_1} V_{(i,z),(j,z')} \go_{i,z}\go_{j,z}=m(\gb)^2\sum_{1\le i,j\le n} V_{(i,S_i),(j,S_j)}\\
+
2m(\gb)\sumtwo{1\le i\le n}{(j,z')\in B_1} V_{(i,S_i),(j,z')}(\go_{j,z'}-m(\gb)\ind_{\{z'=S_j\}})\\
+\sum_{(i,z),(j,z')\in B_1} V_{(i,z),(j,z')}(\go_{i,z}-m(\gb)\ind_{\{z=S_i\}})(\go_{j,z}-m(\gb)\ind_{\{z'=S_j\}}).
\end{multline}
And hence (using the fact that $(x+y)^2\le 2x^2+2y^2$). 
\begin{equation}\label{varcontrol}
 \var_{\bar Q_S} U_1\le 16 m(\gb)^2\sum_{(j,z')\in B_1}\left(\sum_{1\le i\le n}V_{(i,S_i),(j,z')}\right)^2+8\!\!\!\! \sum_{((i,z),(j,z')\in B_1}\!\!\!\! V_{(i,z),(j,z')}^2,
\end{equation}
where we used that $\var_{\bar Q_S} \go_{i,z}\le 2$ (which is true for $\gb$ small enough). The last term is less than $8$ thanks to \eqref{hsnorm}, so that we just have to control the first one.
Independently of the choice of $(j,z')$ we have the bound

\begin{equation}\label{trivbo}
\sum_{1\le i\le n}V_{(i,S_i),(j,z')}\le \frac{\sqrt{\log n}}{C_6 C_7 n}.
\end{equation}
Moreover it is also easy to check that

\begin{equation}
\sum_{(j,z')\in B_1}\sum_{1\le i\le n}V_{(i,S_i),(j,z')}\le \frac{C_7n}{C_6\sqrt{\log n}},
\end{equation}
(these two bounds follow from the definition of $V_{(i,z),(j,z')}$:\ \eqref{vdef}).
Therefore
\begin{equation}
 \sum_{(j,z')\in B_1}\left(\sum_{1\le i\le n}V_{(i,S_i),(j,z')}\right)^2\le \left[\sum_{(j,z')\in B_1}\sum_{1\le i\le n}V_{(i,S_i),(j,z')}\right]\max_{(j,z)\in B_1}\sum_{1\le i\le n}V_{(i,S_i),(j,z')}\le 1.
\end{equation}
Injecting this into \eqref{varcontrol} guaranties that for $\gb$ small enough

\begin{equation}
 \var_{\bar Q_S} U_1 \le 10.
\end{equation}
With Chebyshev inequality, if $K$ has been chosen large enough and 
\begin{equation}
 \bar Q_S U_1\ge 2\exp(K^2),
\end{equation}
we have

\begin{equation}\label{del3}
 \bar Q_S  \left\{ U_1 \le \exp (K^2)\right\}\le \gd/6.
\end{equation}
Hence combining  \eqref{del3} with \eqref{del2} gives
\begin{equation}
 P_x\bar Q_S  \left\{ U_1 \le \exp (K^2);\ \max_{i\in[0,n]}|S_i-S_0|\le (C_6-1)\sqrt{n}\right\}\le \gd/3.
\end{equation}
We use this in \eqref{foor} to get 

\begin{equation}
  \max_{x\in I_0} P_x \bar Q_S \exp\left(f_K\left(U_1\right)\right)\le \frac{\gd}{2}+e^{-K}.
\end{equation}
So that our result is proved provided that $K$ has been chosen large enough.




\end{proof}

\section{Proof of the lower bound in Theorem \ref{pasgaussi}}\label{lb11}

In this section we prove the lower bound for the free-energy in dimension $1$ in arbitrary environment.
To do so we apply the second moment method to some quantity related to the partition function, and combine it
with a percolation argument. The idea of the proof was inspired by a study of a polymer model on hierarchical lattice \cite{LM} where this type of coarse-graining appears naturally.
\medskip
\begin{proposition}\label{th:prop}
There exists a constant $C$ such that for all $\gb\le 1$ we have
\begin{equation}
p(\gb)\ge -C\gb^4 ((\log \gb)^2+1).
\end{equation}
\end{proposition}
\medskip
We use two technical lemmas to prove the result. The first is just a statement about scaling of the random walk, the second is more specific to our problem.

\medskip
\begin{lemma}\label{th:lem1}
There exists an a constant $c_{RW}$ such that for large even squared integers $n$,
\begin{equation}
P\{S_n=\sqrt{n},0<S_i<\sqrt{n} \text{ for } 0<i<n \}= c_{RW} n^{-3/2}+o(n^{-3/2}).
\end{equation}

\end{lemma}
\medskip

\begin{lemma}\label{th:lem2}
For any $\gep>0$ we can find a constant $c_{\gep}$ and $\gb_0$ such that for all $\gb\le \gb_{0}$,
for every even squared integer $n\le c_{\gep}/(\gb^4|\log \gb|)$ we have
\begin{equation}
 \var_Q\left[P \left(\exp\left ( \sum_{i=1}^{n-1}\left(\gb\go_{i,S_i}-\gl(\gb)\right)\right)\ \bigg| \ S_n=\sqrt{n}, 0<S_i<\sqrt{n} \text{ for } 0<i <n \right) \right]<\gep.
\end{equation}
\end{lemma}

\begin{proof}[Proof of Proposition \ref{th:prop} from Lemma \ref{th:lem1} and \ref{th:lem2}]
 
Let $n$ be some fixed integer and define

\begin{equation}
\bar W:= P \exp\left (\sum_{i=1}^{n-1}\left(\gb\go_{i,S_i}-\gl(\gb)\right)\right)\ind_{\{S_n=\sqrt{n}, 0<S_i<\sqrt{n} \text{ for } 0<i <n \}},
\end{equation}
which corresponds to the contribution to the partition function $W_n$ of paths with fixed end point $\sqrt{n}$ staying within a cell of width $\sqrt{n}$, with the specification the environment on the last site is not taken in to account. $\bar W$ depends only of the value of the environment $\go$ in this cell (see figure \ref{fig:restr}).

\begin{figure}[h]
\begin{center}
\leavevmode
\epsfysize =4.5 cm
\psfragscanon
\psfrag{O}[c]{O}
\psfrag{n1}[c]{$n$}
\psfrag{n2}[c]{$\sqrt{n}$}
\epsfbox{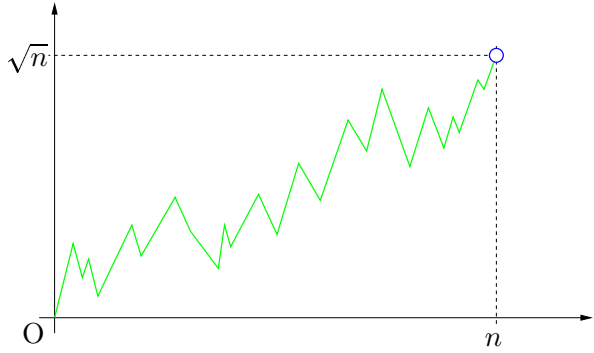}
\caption{\label{fig:restr} We consider a resticted partition function $\bar W$ by considering only paths going from one to the other corner of the cell, without going out. This restriction will give us the independence of random variable corresponding to different cells which will be crucial to make the proof works.}
\end{center}
\end{figure}
One also defines the following quantities for $(i,y)\in \N\times \Z$:

\begin{equation}\begin{split}
 \bar W_i^{(y,y+1)}&:=P_{\sqrt{n}y}\left[ e^{\sum_{k=1}^{n-1}\left[\gb \go_{in+k,S_k}-\gl(\gb)\right]}\ind_{\{S_n=\sqrt{(y+1)n}, 0<S_i-y \sqrt n<\sqrt{n} \text{ for } 0<i <n\}}\right],\\
 \bar W_i^{(y,y-1)}&:=P_{\sqrt{n}y}\left[\ e^{\sum_{k=1}^{n-1}\left[\gb \go_{in+k,S_k}-\gl(\gb)\right]}\ind_{\{S_n=(y-1)\sqrt{n}, -\sqrt{n}<S_i-y\sqrt{n}<0 \text{ for } 0<i <n\}}\right].
\end{split}\end{equation}
which are random variables that have the same law as $\bar W$. Moreover because of independence of the environment in different cells, one can see that
 \begin{equation*}
 \left(\bar W_i^{(y,y\pm 1)};\ (i,y)\in \N\times \Z \text{ such that } i-y \text { is even} \right),
\end{equation*}
is a family of independent variables.

Let $N=nm$ be a large integer. We define $\Omega=\Omega_N$ as the set of path
\begin{equation}
\Omega:=\{ S \ : \ \forall i\in [1,m],\ |S_{in}-S_{(i-1)n}|=\sqrt{n},\ \forall j\in [1,n-1],\  S_{(i-1)n+j}\in \left( S_{(i-1)n},S_{in}\right)\},
\end{equation}
where the interval $\left( S_{i(n-1)},S_{in}\right)$ is to be seen as $\left( S_{in},S_{i(n-1)}\right)$ if $S_{in}<S_{i(n-1)}$, and 
\begin{equation}
 \mathcal S:=\left\{s=(s_0,s_1,\dots,s_m)\in \Z^{m+1}\ :\  s_0=0 \text{ and } |s_i-s_{i-1}|=1, \ \forall i \in[1,m] \right\} .
\end{equation}
We use the trivial bound

\begin{equation}
W_N\ge P \left[\exp\big(\sum_{i=1}^{nm}(\gb\go_{i,S_i}-\gl(\gb))\big)\ind_{\{S\in\Omega\}}\right],
\end{equation}
to get that 

\begin{equation}\label{trtrtr}
 W_N\ge \sum_{s\in \mathcal S} \prod_{i=0}^{m-1} \bar W_{i}^{(s_{i},s_{i+1})}\exp\left(\gb\go_{(i+1)n,s_{i+1}\sqrt{n}}-\gl(\gb)\right).
\end{equation}
(the exponential term is due to the fact the $\bar W$ does not take into account to site in the top corner of each cell).

The idea is of the proof is to find a value of $n$ (depending on $\gb$) such that we are sure that for any value of $m$ we can find a path $s$ such that along the path the values of $(\bar W_{i}^{(s_{i},s_{i+1})})$ are not to low (i.e. close to the expectation of $\bar W$) and to do so, it seems natural to seek for a percolation argument.

Let $p_c$ be the critical exponent for directed percolation in dimension $1+1$ (for an account on directed percolation see \cite[Section 12.8]{perc} and references therein). From Lemma \ref{th:lem2} and Chebyshev inequality, one can find a constant $C_8$ and $\gb_0$ such that for all $n\le\frac{C_8}{\gb^4|\log\gb|}$ and $\gb\le \gb_0$.
\begin{equation}
 Q \{\bar W \ge  Q \bar W /2\}\ge \frac{p_c+1}{2}. \label{eq:perco}
\end{equation}
We choose $n$ to be the biggest squared even integer that is less than $\frac{C_8}{\gb^4|\log\gb|}$.
(in particular have $n\ge \frac{C_8}{2\gb^4|\log\gb|}$ if $\gb$ small enough).

As shown in figure \ref{fig:wnnn}, we associate to our system the following directed percolation picture. For all $(i,y)\in \N\times \Z \text{ such that } i-y$ is even:
\begin{itemize}
 \item If $\bar W_i^{(y,y\pm 1)}\ge (1/2)Q \bar W$, we say that the edge linking the opposite corners of the corresponding cell is open.
 \item If $\bar W_i^{(y,y\pm 1)}< (1/2)Q \bar W$, we say that the same edge is closed.
\end{itemize}
Equation \eqref{eq:perco} and the fact the considered random variables are independent assures that with positive probability there exists an infinite directed path starting from zero.

\begin{figure}[h]
\begin{center}
\leavevmode
\epsfysize =6.5 cm
\psfragscanon
\psfrag{O}[c]{\tiny{O}}
\psfrag{n}[c]{\tiny{$n$}}
\psfrag{2n}[c]{\tiny{$2n$}}
\psfrag{3n}[c]{\tiny{$3n$}}
\psfrag{4n}[c]{\tiny{$4n$}}
\psfrag{5n}[c]{\tiny{$5n$}}
\psfrag{6n}[c]{\tiny{$6n$}}
\psfrag{7n}[c]{\tiny{$7n$}}
\psfrag{8n}[c]{\tiny{$8n$}}
\psfrag{rn}[c]{\tiny{$+\sqrt{n}$}}
\psfrag{r2n}[c]{\tiny{$+2\sqrt{n}$}}
\psfrag{r3n}[c]{\tiny{$+3\sqrt{n}$}}
\psfrag{r4n}[c]{\tiny{$+4\sqrt{n}$}}
\psfrag{m1n}[c]{\tiny{$-\sqrt{n}$}}
\psfrag{m2n}[c]{\tiny{$-2\sqrt n$}}
\psfrag{m3n}[c]{\tiny{$-3\sqrt n$}}
\psfrag{1}[c]{\tiny{$1$}}
\psfrag{2}[c]{\tiny{$2$}}
\psfrag{3}[c]{\tiny{$3$}}
\psfrag{4}[c]{\tiny{$4$}}
\psfrag{5}[c]{\tiny{$5$}}
\psfrag{6}[c]{\tiny{$6$}}
\psfrag{7}[c]{\tiny{$7$}}
\psfrag{8}[c]{\tiny{$8$}}
\epsfbox{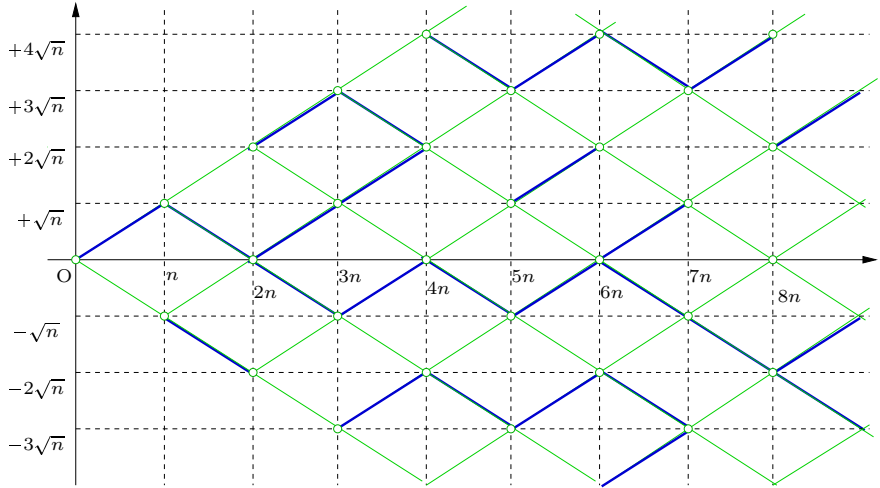}
\end{center}
\caption{\label{fig:wnnn} This figure illustrates the percolation argument used in the proof. To each cell is naturally associated a random variable $\bar W_i^{y,y\pm 1}$, and these  random variables are i.i.d. When $\bar W_i^{y,y\pm 1}\ge 1/2 Q \bar W$ we open the edge in the corresponding cell (thick edges on the picture). As this happens with a probability strictly superior to $p_c$, we have a positive probability to have an infinite path linking $0$ to infinity.}
\end{figure}

When there exists an infinite open path is linking zero to infinity exists, we can define the highest open path in an obvious way. Let $(s_i)_{i=1}^m$ denotes this highest path. If $m$ is large enough, by law of large numbers we have that with a probability close to one,
\begin{equation}
\sum_{i=1}^m \left[\gb\go_{ni,\sqrt{n}s_i}-\gl(\gb)\right] \ge -2m\gl(\gb). 
\end{equation}

Using this and and the percolation argument with \eqref{trtrtr} we finally get that with a positive probability which does not depend on $m$ we have

\begin{equation}
 W_{nm}\ge \left[(1/2)e^{-2\gl(\gb)}Q\bar W\right]^m.
\end{equation}
Taking the $\log$ and making $m$ tend to infinity this implies that

\begin{equation}
 p(\gb)\ge \frac{1}{n}\left[-2 \gl(\gb)-\log 2+ \log Q \bar W\right]\ge -\frac{c}{n}\log n.
\end{equation}
For some constant $c$, if $n$ is large enough (we used Lemma \ref{th:lem1} to get the last inequality.
The result follows by replacing $n$ by its value.




\end{proof}
\medskip

\begin{proof}[Proof of Lemma \ref{th:lem1}]

Let $n$ be square and even. $T_k$, $k\in \Z$ denote the first hitting time of $k$ by the random walk $S$ (when $k=0$ it denotes the return time to zero). 
We have
\begin{multline}
P\{S_n=\sqrt{n},0<S_i<\sqrt{n}, \text{ for all } 1<i<n\}\\
=\sum_{k=1}^{n-1} P\{T_{\sqrt{n}/2}=k,\ S_j>0 \text{ for all } j<n \text{ and } T_{\sqrt{n}}=n\}\\
= P\{T_{\sqrt{n}/2}<n,\ S_j<\sqrt{n} \text{ for all } j<n \text{ and } T_{0}=n)\},
\end{multline}
where the second equality is obtained with the strong Markov property used for $T=T_{\sqrt{n}/2}$, and the reflexion principle for the random walk.
The last line is equal to
\begin{equation}\label{eq:hop}
P\{\max_{k\in [0,n]} S_k\in [\sqrt{n}/2,\sqrt{n})|T_0=n\}P\{T_0=n\}.
\end{equation}
We use here a variant of Donsker's Theorem, for a proof see \cite[Theorem 2.6]{K}.
\medskip

\begin{lemma}
The process
\begin{equation}
t\mapsto \left\{\frac {S_{\lceil nt \rceil }}{\sqrt{n}}\ \bigg| \ T_0=n\right\}, \quad t\in[0,1]
\end{equation}
converges in distribution to the normalized Brownian excursion in the space $D([0,1],\R)$.
\end{lemma}
\medskip
We also know that (see for example \cite[Proposition A.10]{Book}) for $n$ even $P(T_0=n)=\sqrt{2/\pi}n^{-3/2}+o(n^{-3/2})$.
Therefore, from  \eqref{eq:hop} we have
\begin{multline}
P\{S_n=\sqrt{n},0<S_i<\sqrt{n}, \text{ for all } 1<i<n\}\\
= \sqrt{2/\pi}n^{-3/2}\bbP\left[\max_{t\in[0,1]} {\bf e}_t\in (1/2,1)\right]+o(n^{-3/2}).
\end{multline}
Where $\bf e$ denotes the  normalized Brownian excursion, and $\bbP$ its law.
\end{proof}
\medskip

\begin{proof}[Proof of Lemma \ref{th:lem2}]

Let $\gb$ be fixed and  small enough, and $n$ be some squared even integer which is less than $c_\gep/(\gb^4|\log \gb|)$. We will fix the value $c_{\gep}$ independently of $\gb$ later in the proof, and always consider that $\gb$ is sufficiently small.
By a direct computation the variance of

\begin{equation}
P\left[\exp\left( \sum_{i=1}^{n-1}[\gb\go_{i,S_i}-\gl(\gb)]\right)\bigg| \ S_n=\sqrt{n}, 0<S_i<\sqrt{n} \text{ for } 0<i <n \right] 
\end{equation}
is equal to

\begin{equation}\label{eq:abcd}
 P^{\otimes 2}\left[\exp \left(\sum_{i=1}^{n-1}\gga(\gb)\ind_{\{S^{(1)}_i=S^{(2)}_i\}}\right)\ \bigg| \ A_n\right]-1.
\end{equation}
where

\begin{equation}
 A_n=\left\{S_n^{(j)}=\sqrt{n}, 0<S^{(j)}_i<\sqrt{n} \text{ for } 0<i<n,\ j=1,2\right\},
\end{equation}
and $\gga(\gb)=\gl(2\gb)-2\gl(\gb)$ (recall that $\gl(\gb)=\log Q \exp(\gb\go_{(1,0)})$), and $S_n^{(j)}$, $j=1,2$ denotes two independent random walk with law denoted by $P^{\otimes 2}$.
From this it follows that if $n$ is small the result is quite straight--forward. We will therefore only be interested in the case of large $n$ (i.e.\ bounded away from zero by a large fixed constant).

We define $\tau=\left(\tau_k\right)_{k\ge 0}=\{S^{(1)}_i=S^{(2)}_i,i\ge 0\}$  the set where the walks meet (it can be written as an increasing sequence of integers). By the Markov property, the random variables $\tau_{k+1}-\tau_k$ are i.i.d.\ , we say that $\tau$ is a {\sl renewal sequence}.
\medskip

We want to bound the probability that the renewal sequence $\tau$ has too many returns before times $n-1$, in order to estimate \eqref{eq:abcd}. To do so, we make the usual computations with Laplace transform.

From \cite[p. 577]{Feller} , we know that
\begin{equation}\label{Lapl}
1- P^{\otimes 2} \exp(- x \tau_1)= \frac{1}{\sum_{n\in \N} \exp(-xn)P\{S_n^{(1)}=S_n^{(2)}\}}.
\end{equation}
Thanks to the the local central limit theorem for the simple random walk, we know that for large $n$
\begin{equation}
P\{S_n^{(1)}=S_n^{(2)}\}= \frac{1}{\sqrt{\pi n}}+o(n^{-1/2}). 
\end{equation}
So that we can get from \eqref{Lapl} that when $x$ is close to zero
\begin{equation}
 \log P^{\otimes 2} \exp(- x \tau_1)=  -\sqrt{x} +o(\sqrt{x}).
\end{equation}
We fix $x_0$ such that $\log P \exp(- x \tau_1)\le \sqrt x /2$ for all $x\le x_0$.
For any $k\le n$ we have

\begin{equation} \begin{split}
 P^{\otimes 2}\{|\tau\cap [1,n-1]|\ge k\}=P^{\otimes 2}\{\tau_k\le n-1\}&\le \exp((n-1)x) P^{\otimes 2} \exp (-\tau_k x)\\
&\le \exp\left[nx+k \log P^{\otimes 2}\exp(-x\tau_1)\right]. \end{split}
\end{equation}
For any $k\le \left\lfloor 4 n\sqrt{x_0}\right\rfloor=k_0$ one can choose $x=(k/4 n)^2\le x_0$ in the above and use the definition of $x_0$ to get that
 
\begin{equation}
  P^{\otimes 2}\{|\tau\cap [1,n-1]|\ge k\}\le \exp\left(-k^2/(32  n)\right).
\end{equation}
In the case where $k> k_0$ we simply bound the quantity by

\begin{equation}
   P^{\otimes 2}\{|\tau\cap [1,n-1]|\ge k\}\le \exp\left(k_0^2/(32 n)\right)\le \exp\left(-n x_0/4\right).
\end{equation}
By Lemma \eqref{th:lem1}, if $n$ is large enough
\begin{equation}
 P^{\otimes}A_n\ge 1/2 c_{RW}^2n^{-3}.
\end{equation}
A trivial bound on the conditioning gives us

\begin{equation}\label{bounee}\begin{split}
P^{\otimes 2}\left(|\tau\cap [1,n-1]|\ge k \ \big| \ A_n \right)&\le\min\left(1,2c_{RW}^{-2} n^3\exp\left(-k^2/(32 n)\right)\right) \text{ if } k\le k_0,\\
 P^{\otimes 2}\left(|\tau\cap [1,n-1]|\ge k \ \big| \ A_n \right)&\le 2c_{RW}^{-2}n^{3} \exp\left(-n x_0/4\right)
\text{ otherwise}.
\end{split}\end{equation}
We define $k_1:=\lceil 16\pi\sqrt{n\log(2c_{RW}^{-2} n^3)}\rceil$. 
The above implies that for $n$ large enough we have

\begin{equation}\begin{split}
P^{\otimes 2}\left(|\tau\cap [1,n-1]|\ge k \ \big| \ A_n \right)&\le 1 \text{ if } k\le k_1, \\
P^{\otimes 2}\left(|\tau\cap [1,n-1]|\ge k \ \big| \ A_n \right)&\le \exp\left(-k^2/(64 n)\right) \text{ if } k_1\le k\le k_0,\\
 P^{\otimes 2}\left(|\tau\cap [1,n-1]|\ge k \ \big| \ A_n \right)&\le\exp\left(-n x_0/8\right)
\text{ otherwise}.
\end{split}\end{equation}
Now we are ready to bound \eqref{eq:abcd}.
Integration by part gives,

\begin{equation}\begin{split}
 &P^{\otimes 2}\left[\exp\left(\gamma \gb |\tau\cap [1,n-1]|\right)\ \big| \ A_n\right]-1\\
&\quad\quad\quad\quad\quad\quad\quad\quad=\gga(\gb)\int_{0}^{\infty}\exp(\gga(\gb) x)P^{\otimes 2}\left(|\tau\cap [1,n-1]|\ge x \ \big| \ A_n\right)\dd x.
\end{split}\end{equation}
We split the right-hand side in three part corresponding to the three different bounds we have in \eqref{bounee}:
$x\in[0,k_1]$, $x\in[k_1,k_0]$ and $x\in[k_0,n]$. It suffices to show that each part is less than $\gep/3$ to finish the proof. The first part is

\begin{equation}
\gga(\gb)\int_{0}^{k_1}\exp(\gga(\gb) x)P^{\otimes 2}\left(|\tau\cap [1,n-1]|\ge x\ \big| \ A_n\right)\dd x \le \gga(\gb)k_1\exp(\gga(\gb)k_1).
\end{equation}
One uses that $n\le \frac{c_{\gep}}{\gb^4|\log \gb|}$ and $\gga(\gb)=\gb^2+o(\gb^2)$ to get that for $\gb$ small enough and $n$ large enough if $c_\gep$ is well chosen we have
\begin{equation}
 k_1\gga(\gb)\le 100 \gb^2 \sqrt{n\log n}\le \gep/4,
\end{equation}
so that $\gga(\gb)k_1\exp(\gga(\gb)k_1)\le \gep/3$.\\
We use our bound for the second part of the integral to get

\begin{equation}\begin{split}
&\gga(\gb)\int_{k_1}^{k_0}\exp(\gga(\gb) x)P^{\otimes 2}\left(|\tau\cap [1,n-1]|\ge x\ \big| \ A_n\right)\dd x\\
&\quad\le \gga(\gb)\int_{0}^{\infty}\exp\left(\gga(\gb)x-x^2/(64 n)\right)\dd x =\int_{0}^{\infty}\exp\left(x-\frac{x^2}{64 n \gga(\gb)^2}\right)\dd x.
\end{split}\end{equation}
Replacing $n$ by its value, we see that the term that goes with $x^2$ in the exponential can be made arbitrarily large, provided that $c_\gep$ is small enough. In particular we can make the left-hand side less than $\gep/3$.\\
Finally, we estimate the last part

\begin{equation}\begin{split}
&\gga(\gb)\int_{k_0}^{n}\exp(\gga(\gb)^2 x)P^{\otimes 2}\left(|\tau\cap [1,n-1]|\ge x \ \big| \ A_n\right)\dd x\\
&\quad\quad\quad\quad\quad\quad\quad\quad \le \gga(\gb)\int_{0}^{n}\exp(\gga(\gb)x-n x_0/8)\dd x=n\exp(-[\gga(\gb)-x_0/8]n).
\end{split}\end{equation}
This is clearly less than $\gep/3$ if $n$ is large and $\gb$ is small.

\end{proof}

\section{Proof of the lower bound of Theorem \ref{gaussi}}\label{dim11}

In this section we use the method of replica-coupling that is used for the disordered pinning model in \cite{F_rc} to derive a lower bound on the free energy. The proof here is an adaptation of the argument used there to prove disorder irrelevance.

The main idea is the following: Let $W_N(\gb)$ denotes the renormalized partition function for inverse temperature $\gb$. A simple Gaussian computation gives

\begin{equation}
 \frac{\dd Q \log W_N(\sqrt{t})}{\dd t}\bigg|_{t=0}=-\frac{1}{2}P^{\otimes 2}\sum_{i=1}^N \ind_{\{S_i^{(1)}=S_i^{(2)}\}}.
\end{equation}
Where $S^{(1)}$ and $S^{(2)}$ are two independent random walk under the law $P^{\otimes 2}$.
This implies that for small values of $\gb$ (by the equality of derivative at $t=0$),

\begin{equation}\label{eqqq}
Q \log W_N(\gb) \approx - \log P^{\otimes 2} \exp\left(\gb^2/2\sum_{i=1}^N \ind_{\{S_N^{(1)}=S_N^{(2)}\}}\right).
\end{equation}
This tends to make us believe that

\begin{equation}
 p(\gb)=-\lim_{N\rightarrow \infty} \log P^{\otimes 2} \exp\left(\gb^2/2\sum_{i=1}^N \ind_{\{S_N^{(1)}=S_N^{(2)}\}}\right).
\end{equation}
However, things are not that simple because \eqref{eqqq} is only valid for fixed $N$, and one needs some more work to get something valid when $N$ tends to infinity.
The proofs aims to use convexity argument and simple inequalities to be able to get the inequality

\begin{equation}
 p(\gb)\ge-\lim_{N\rightarrow \infty} \log P^{\otimes 2} \exp\left(2 \gb^2\sum_{i=1}^N \ind_{\{S_N^{(1)}=S_N^{(2)}\}}\right).
\end{equation}
The fact that convexity is used in a crucial way make it quite hopeless to get the other inequality using this method.

\begin{proof}
Let use define for $\gb$ fixed and $t\in[0,1]$

\begin{equation}
\Phi_N(t,\gb):=\frac{1}{N}Q\log P \exp\left(\sum_{i=1}^N \left[\sqrt{t}\gb \go_{i,S_i}-\frac{t\gb^2}{2}\right]\right),
 \end{equation}
and for $\gl\ge 0$
\begin{equation}
\Psi_N(t,\gl,\gb):=\frac{1}{2N}Q\log P^{\otimes 2} \exp\left(\sum_{i=1}^N\left[\sqrt{t}\gb(\go_{i,S_i^{(1)}}+\go_{i,S_i^{(2)}})-t\gb^2+\gl\gb^2\ind_{\{S_i^{(1)}=S_i^{(2)}\}}\right]\right).
\end{equation}
One can notice that $\Phi_N(0,\gb)=0$ and $\Phi_N(1,\gb)=p_N(\gb)$ (recall the definition of $p_N$ \eqref{Pn}), so that  $\Phi_N$ is an interpolation function.
Via the Gaussian integration by par formula

\begin{equation}
 Q \go f(\go)=Q f'(\go),
\end{equation}
valid (if $\go$ is a centered standard Gaussian variable) for every differentiable functions such that $\lim_{|x|\to\infty}\exp(-x^2/2)f(x)=0$, one finds

\begin{equation}\begin{split}\label{fifi}
 \frac{\dd }{\dd t}\Phi_N(t,\gb)&=-\frac{\gb^2}{2N}\sum_{j=1}^{N}\sum_{z\in \Z} Q \left(\frac{ P \exp\left(\sum_{i=1}^N \left[\sqrt{t}\gb\go_{i,S_i}-\frac{t\gb^2}{2}\right]\right)\ind_{\left\{S_j=z\right\}}}{P\exp\left(\sum_{i=1}^N \left[\sqrt{t}\gb\go_{i,S_i}-\frac{t\gb^2}{2}\right]\right)}\right)^2 \\
&=-\frac{\gb^2}{2N}Q \left(\mu^{(\sqrt{t}\gb)}_n\right)^{\otimes 2} \left[\sum_{i=1}^N \ind_{\{S_i^{(1)}=S_i^{(2)}\}}\right].
\end{split}
\end{equation}
This is (up to the negative multiplicative constant $-\gb^2/2$) the expected overlap fraction of two independent replicas of the random--walk under the polymer measure for the inverse temperature $\sqrt{t}\gb$. This result has been using It\^o formula in \cite[Section 7]{CH_ptrf}.\\
For notational convenience, we define

\begin{equation}
H_N(t,\gl,S^{(1)},S^{(2)})=\sum_{i=1}^N\left[\sqrt{t}\gb(\go_{i,S_i^{(1)}}+\go_{i,S_i^{(2)}})-t\gb^2+\gl\gb^2\ind_     {\left\{S_i^{(1)}=S_i^{(2)}\right\}}\right].
\end{equation}
We use Gaussian integration by part again, for $\Psi_N$: 

\begin{multline}\label{psipsi}
\frac{\dd}{\dd t}\Psi_{N}(t,\gl,\gb)= \frac{\gb^2}{2N}\sum_{j=1}^N Q \frac{ P^{\otimes 2} \exp\left(H_N(t,\gl,S^{(1)},S^{(2)})\right)\ind_{\{S_j^{(1)}=S_j^{(2)}\}}}{P^{\otimes 2} \exp\left(H_N(t,\gl,S^{(1)},S^{(2)})\right)}\\
- \frac{\gb^2}{4N}\sum_{j=1}^N\sum_{z\in \Z} Q \left(\frac{ P^{\otimes 2} \left(\ind_{\{S_j^{(1)}=z\}}+\ind_{\{S_j^{(2)}=z\}}\right)\exp\left(H_N(t,\gl,S^{(1)},S^{(2)})\right)}{ P^{\otimes 2} \exp\left(H_N(t,\gl,S^{(1)},S^{(2)})\right)}\right)^2 \\
\le \frac{\gb^2}{2N}\sum_{j=1}^N Q \frac{ P^{\otimes 2} \exp\left(H_N(t,\gl,S^{(1)},S^{(2)})\right)\ind_{\{S_j^{(1)}=S_j^{(2)}\}}}{P^{\otimes 2} \exp\left(H_N(t,\gl,S^{(1)},S^{(2)})\right)}=\frac{\dd}{\dd\gl} \Psi_N(t,\gl,\gb).
\end{multline}
The above implies that for every $t\in[0,1]$ and $\gl\ge 0$

\begin{equation}\label{hmhm}
 \Psi_N(t,\gl,\gb)\le \Psi_N(0,\gl+t,\gb).
\end{equation}
Comparing \eqref{fifi} and \eqref{psipsi}, and using convexity and monotonicity of $\Psi_N(t,\gl,\gb)$ with respect to $\gl$, and the fact that $\Psi_N(t,0,\gb)=\Phi_N(t,\gb)$ one gets

\begin{multline}
                  -\frac{\dd}{\dd t}\phi_N(t,\gb)=\frac{\dd}{\dd \gl}\Psi_N(t,\gl,\gb)\bigg|_{\gl=0}\\
\le \frac{\Psi_N(t,2-t,\gb)-\Phi_N(t,\gb)}{2-t}\le \Psi_N(0,2,\gb)-\Phi_N(t,\gb),
                \end{multline}
where in the last inequality we used $(2-t)\ge 1$ and \eqref{hmhm}.
Integrating this inequality between $0$ and $1$ and recalling $\Phi_N(1,\gb)=p_N(\gb)$ we get

\begin{equation}
 p_N(\gb)\ge (1-e)\Psi_N(0,2,\gb).
\end{equation}

On the right-hand side of the above we recognize something related to pinning models. More precisely

\begin{equation}\label{alm}
 \Psi_N(0,2,\gb)=\frac{1}{2N}\log Y_N,
\end{equation}
where 
\begin{equation}
 Y_N=P^{\otimes 2} \exp\left(2\gb^2 \sum_{i=1}^N \ind_{\left\{S_N^{(1)}=S_N^{(2)}\right\}}\right)
\end{equation}
is the partition function of a homogeneous pinning system of size $N$ and parameter $2\gb^2$ with underlying renewal process the sets of zero of the random walk $S^{(2)}-S^{(1)}$. This is a well known result in the study of pinning model ( we refer to \cite[Section 1.2]{Book} for an overview and the results we cite here) that 

\begin{equation}
 \lim_{N\rightarrow\infty} \frac{1}{N}\log Y_N=\tf(2\gb^2),
\end{equation}
where $\tf$ denotes the free energy of the pinning model. Moreover, it is also stated 
\begin{equation}
 \tf(h)\stackrel{h\to 0+}{\sim} h^2/2.
\end{equation}
Then passing to the limit in \eqref{alm} ends the proof of the result for any constant strictly bigger that $4$. 
\end{proof}

\section{Proof the lower bound in Theorem \ref{1+2uplb}}\label{dim12}

The technique used in the two previous sections could be adapted here to prove the results but in fact it is not necessary.
Because of the nature of the bound we want to prove in dimension $2$ (we do not really track the best possible constant in the exponential), it will be sufficient here to control the variance of $W_n$ up to some value, and then the concentration properties of $\log W_n$ to get the result. The reader can check than using the same method in dimension $1$ does not give the right power of $\gb$.

First we prove a technical result to control the variance of $W_n$ which is the analog of \eqref{th:lem2} in dimension $1$. Recall that $\gga(\gb):=\gl(2\gb)-2\gl(\gb)$ with $\gl(\gb):=\log Q \exp (\gb\go_{(1,0)})$.

\begin{lemma}\label{th:lemmvar}
For any $\gep<0$, one can find a constant $c_\gep>0$ and $\gb_0>0$ such that for any $\gb\le \gb_0$, for any 
$n\le \exp\left(c_{\gep}/\gb^2\right)$ we have
\begin{equation}
\var_Q W_n\le \gep.
\end{equation}
\end{lemma}
\begin{proof}
A straight--forward computation shows that the the variance of $W_n$ is given by

\begin{equation}\label{var2d}
\var_Q W_n= P^{\otimes 2} \exp\left(\gga(\gb)\sum_{i=1}^n \ind_{\{S_i^{(1)}=S_i^{(2)}\}}\right)-1.
\end{equation}
where $S^{(i)}$, $i=1,2$ are two independent $2$--dimensional random walks.

As the above quantity is increasing in $n$, it will be enough to prove the result for $n$ large. For technical convenience we choose to prove the result for $n= \rfloor \exp(-c_{\gep}/\gga(\gb))\rfloor$ (recall $\gga(\gb)=\gl(2\gb)-2\gl(\gb)$) which does not change the result since $\gga(\gb)=\gb^2+o(\gb^2)$.

The result we want to prove seems natural since we know that $(\sum_{i=1}^n \ind_{\{S_i^{(1)}=S_i^{(2)}\}})/\log n$ converges to an exponential variable (see e.g.\ \cite{GZ}), and $\gga(\gb)\sim c_{\gep}\log n$. However, convergence of the right--hand side of \eqref{var2d} requires the use of the dominated convergence Theorem, and the proof of the domination hypothesis is not straightforward. It could be extracted from the proof of the large deviation result in \cite{GZ}, however
we include a full proof of convergence here for the sake of completeness.

We define $\tau=\left(\tau_k\right)_{k\ge 0}=\{S^{(1)}_i=S^{(2)}_i,i\ge 0\}$  the set where the walks meet (it can be written as an increasing sequence). By the Markov property, the random variables $\tau_{k+1}-\tau_k$ are i.i.d.\ . 

\medskip

To prove the result, we compute bounds on the probability of having too many point before $n$ in the renewal $\tau$. As in the $1$ dimensional case, we use Laplace transform to do so.
From \cite[p. 577]{Feller} , we know that

\begin{equation}\label{Lapl1}
1- P^{\otimes 2} \exp(- x \tau_1)= \frac{1}{\sum_{n\in \N} \exp(-xn)P\{S_n^{(1)}=S_n^{(2)}\}}.
\end{equation}
The local central limit theorem says that for large $n$

\begin{equation}
  P^{\otimes 2}\{S_n^{(1)}=S_n^{(2)}\}\sim \frac{1}{\pi n}.
\end{equation}
Using this into \eqref{Lapl1} we get that when $x$ is close to zero

\begin{equation}
\log P^{\otimes 2} \exp(- x \tau_1)\sim -\frac{\pi}{|\log x|}.
\end{equation}
We use the following estimate 

\begin{equation} \label{coucou}\begin{split}
 P^{\otimes 2}\{|\tau\cap [1,n]|\ge k\}=P^{\otimes 2}\{\tau_k\le n\}&\le \exp(nx) P^{\otimes 2} \exp (-\tau_k x)\\
&= \exp\left[nx+k\log  P^{\otimes 2} \exp(-x\tau_1)\right]. \end{split}
\end{equation}
Let $x_0$ be such that for any $x\le x_0$, $\log P^{\otimes 2} \exp(- x \tau_1)\ge -3/ |\log x|$.
For $k$ such that $k/(n\log(n/k))\le x_0$, we replace $x$ by $k/(n\log(n/k))$ in \eqref{coucou}
to get 
\begin{equation}
  P^{\otimes 2}\{|\tau\cap [1,n]|\ge k\}\le \exp \left(\frac{k}{\log(n/k)}-\frac{3 k}{\log\left[k/(n\log n/k)\right]}\right)\le \exp\left(-\frac{k}{\log(n/k)}\right),
\end{equation}
where the last inequality holds if $k/n$ is small enough. We fix $k_0=\gd n$ for some small $\gd$.
We get that
\begin{equation}\label{2db}\begin{split}
   P^{\otimes 2}\{|\tau\cap [1,n]|\ge k\}&\le \exp\left(-\frac{k}{\log(n/k)}\right) \quad \text{if } k\le k_0\\
   P^{\otimes 2}\{|\tau\cap [1,n]|\ge k\}&\le \exp\left(-\frac{k_0}{\log(n/k_0)}\right)=\exp\left(-\frac{\gd n}{\log (1/\gd)}\right) \quad \text{if } k\ge k_0.
\end{split}
\end{equation}
We are ready to bound \eqref{var2d}. We remark that using integration by part we obtain

\begin{equation}
P \exp\left(\gga(\gb) |\tau\cap [1,n]|\right)-1= \int_{0}^n \gga(\gb)\exp(\gga(\gb)x)P^{\otimes 2}(\tau\cap [1,n]|\ge x)\dd x.
\end{equation}
To bound the right--hand side, we use the bounds we have concerning $\tau$: \eqref{2db}. We have to split the integral in three parts. 
\medskip

The integral between $0$ and $1$ can easily be made less than $\gep/3$ by choosing $\gb$ small.

\medskip
Using $n\le \exp(c_{\gep}/\gga(\gb))$, we get that

\begin{multline}
 \int_{1}^{\gd n} \gga(\gb)\exp(\gga(\gb)x)P^{\otimes 2}(\tau\cap [1,n]|\ge x)\dd x\le \int_{1}^{\gd n}\gga(\gb)\exp\left(\gga(\gb)x-\frac{x}{\log(n/x)}\right)\dd x\\
\le \int_{1}^{\gd n}\gga(\gb)\exp\left(\gga(\gb)x-\frac{\gga(\gb)\gb x}{c_{\gep}}\right)\le \frac{c_{\gep}}{1-c_{\gep}}.
\end{multline}
This is less that $\gep/3$ if $c_{\gep}$ is chosen appropriately.
The last part to bound is
\begin{equation}
 \int_{\gd n}^n \gga(\gb)\exp(\gga(\gb)x)P^{\otimes 2}(\tau\cap [1,n]|\ge x)\le n\gga(\gb)\exp\left(\gga(\gb)n-\frac{\gd n}{\log 1/\gd}\right)\le \gep/3,
\end{equation}
where the last inequality holds if $n$ is large enough, and $\gb$ is small enough.
 
\end{proof}

\begin{proof}[Proof of the lower bound in Theorem \ref{1+2uplb}]

By a martingale method that one can find a constant $c_9$ such that
\begin{equation}
\var_Q{\log W_n}\le C_9 n, \quad \quad \forall n\ge 0, \forall \gb\le 1.
\end{equation}
(See \cite[Proposition 2.5]{CSY} and its proof for more details).\\
 Therefore Chebyshev inequality gives
\begin{equation}\label{eq:ineq}
Q\left\{\left|\frac 1 n \log W_n - \frac 1 n Q\log W_n\right|\ge n^{-1/4}\right\}\le C_9 n^{-1/2}.
\end{equation}
Using  Lemma \ref{th:lemmvar} and Chebyshev inequality again, we can find a constant $C_{10}$ such that for small $\gb$ and $n=\lceil\exp(C_{10}/\gb^2)\rceil$ we have
\begin{equation}
Q\left\{ W_n<1/2\right\}\le 1/2.
\end{equation}
This combined with \eqref{eq:ineq} implies that

\begin{equation}
\frac{-\log 2}{n}\le  n^{-1/4} + Q \frac 1 n \log W_n\le n^{-1/4}+p(\gb).
\end{equation}
Replacing $n$ by its value we get

\begin{equation}
 p(\gb)\ge  -n^{-1/4}-\frac{\log 2}{n}\ge -\exp (-C_{10}/5\gb^2).
\end{equation}

\end{proof}

\bigskip
{\bf Acknowledgements:} The author is very grateful to Giambattista Giacomin for numerous suggestions and precious help for the writing of this paper,  to Francesco Caravenna for the proof of Lemma \ref{th:lem1} and to Fabio Toninelli and Francis Comets for enlightening discussions. The author also acknowledges the support of ANR, grant POLINTBIO.


\begin{thebibliography}{99}
\bibitem{AZ} S. Albeverio. and X. Zhou, {\it A martingale approach to directed polymers in a random environment}, J. Theoret. Probab. {\bf 9} (1996) 171--189.

\bibitem{Bertin} P. Bertin, \emph{Free energy for Linear Stochastic Evolutions in
dimension two}, preprint (2009).


\bibitem{B} E. Bolthausen, {\it A note on diffusion of directed polymer in a random environment}, Commun. Math. Phys. {\bf 123} (1989) 529--534.

\bibitem{CH_ptrf} P. Carmona and Y. Hu,
{\sl On the partition function of a directed polymer in a random Gaussian environment} ,
Probab. Theor. Relat. Fields {\bf 124} 3 (2002) 431-457. 

\bibitem{CH_al} P. Carmona and Y. Hu, {\sl Strong disorder implies strong localization for directed
polymers in a random environment}, ALEA {\bf 2} (2006) 217--229.

\bibitem{CSY} F. Comets, T. Shiga and N. Yoshida, {\it Directed Polymers in a
   random environment: strong disorder and path localization}, Bernouilli {\bf 9} 4 (2003) 705-723.

\bibitem{CSY_rev} F. Comets, T. Shiga, and N. Yoshida, {\sl Probabilistic Analysis of Directed Polymers in a Random Environment: a Review }, Adv. Stud. Pure Math. {\bf 39} (2004) 115-142.

\bibitem{CV} F. Comets and V. Vargas, {\sl Majorizing multiplicative cascades for directed polymers in random media},  ALEA  {\bf 2} (2006) 267--277. 

\bibitem{CY} F. Comets and N. Yoshida, {\it Directed polymers in a random environment
   are diffusive at weak disorder}, Ann. Probab. {\bf 34} 5 (2006) 1746-1770.


\bibitem{cf:DGLT}
B. Derrida,
G. Giacomin, H. Lacoin and F. L. Toninelli, {\it
Fractional moment bounds and disorder relevance for pinning models},
Commun. Math. Phys. {\bf 287} (2009) 867--887.

\bibitem{Feller} 
Feller W., {\it An Introduction to Probability Theory and Its Applications, Volume II}, John Wiley \& Sons, Inc. New York (1966). 

\bibitem{GZ} N. Gantert and O. Zeitouni, \emph{Large and moderate deviations for local time of a reccurent Markov chain on {\bf $Z^2$}}, Ann. Inst. H. Poincaré Probab. Statist. {\bf 34} (1998), 687--704.

\bibitem{Book} G.~Giacomin, {\it Random polymer models}, IC press,
World Scientific, London (2007).

\bibitem{cf:GLT} G. Giacomin, H. Lacoin and F. L. Toninelli, \emph{
    Hierarchical pinning models, quadratic maps and quenched
    disorder}, To appear in Probab. Theory. Rel. Fields, arXiv:0711.4649 [math.PR].

\bibitem{GLTm} G. Giacomin, H. Lacoin and F.L. Toninelli, {\it Marginal relevance of disorder for pinning models}, to appear in Commun. Pure Appl. Math. , arXiv:0811.0723 [math-ph].

\bibitem{cf:kbodies} 
G. Giacomin, H. Lacoin and F. L. Toninelli,
\textit{Disorder relevance at marginality and critical point shift}, preprint (2009) arXiv:0906.1942v1 [math-ph].


\bibitem{perc} G. Grimmett, {\it Percolation} Second Edition, Grundlehren der Mathematischen Wissenschaften {\bf 321}, Springer-Verlag, Berlin (1999).


\bibitem{HH} D.A. Huse and C.L. Henley, {\it Pinning and roughening of domain wall in Ising systems due to
random impurities}, Phys. Rev. Lett. {\bf 54} (1985) 2708--2711.

\bibitem{IS} J.Z. Imbrie and T. Spencer, {\it Diffusion of directed polymer in a random environment}, J. Stat. Phys. {\bf 52} 3/4 (1988), 608--626.

\bibitem{K}  W.D. Kaigh,  {\it An invariance principle for random walk conditioned by a late return to zero}, 
Ann. Probab. {\bf 4} (1976) 115-121 .

   
\bibitem{LM} H. Lacoin and G. Moreno {\it Directed polymer on hierarchical lattice with site disorder}, preprint (2009) arXiv:	arXiv:0906.0992v1 [math.PR]. 

\bibitem{LV} Q. Liu and F. Watbled {\it Exponential inequalities for martingales and asymptotic properties of the free energy of directed polymers in random environment}, preprint (2008) arXiv:0812.1719v1.

\bibitem{SZ} R. Song and X.Y. Zhou, {\it A remark on diffusion of directed polymers in random environment}, J. Stat. Phys. {\bf 85} 1/2 (1996) 277--289.

\bibitem{F_rc} F.L. Toninelli, {\sl A replica coupling-approach to disordered pining models} Commun. Math. Phys.  {\bf 280} (2008) 389-401.

\bibitem{F} F.L. Toninelli, {\it Coarse graining, fractional moments and the critical slope of random copolymers},  To appear in Electron. Journal Probab. arXiv:0806.0365 [math.PR].
   





\bibitem{V} V. Vargas, {\it Strong localization and macroscopic atoms for directed polymer}, Probab. Theor. Relat. Fields, {\bf 134} 3/4 (2008) 391--410.


\end{thebibliography}
\end{document}